\documentclass[a4paper,10pt]{article}

\usepackage{multirow}
\usepackage{subfigure, amsfonts, amsmath}
\usepackage{longtable}
\usepackage{pdflscape}
\usepackage{float}
\usepackage{cancel}
\usepackage{setspace}
\usepackage{graphicx}
\usepackage{array}
\usepackage{hyperref}
\usepackage{appendix}
\usepackage{amssymb}
\usepackage{epstopdf}
 \usepackage{multirow}
\usepackage{epsfig}
\usepackage{subfigure, amsfonts, amsmath}
\usepackage{longtable}
\usepackage{lscape}
\usepackage{subfigure}
\usepackage{amssymb}
\usepackage{a4,fancyhdr}
\usepackage{amsfonts}
\usepackage{amsmath}
\usepackage{yhmath}
\usepackage[english]{babel}
\usepackage[latin1]{inputenc}
\usepackage[T1]{fontenc}
\usepackage{enumerate}
\usepackage{tikz}
\usetikzlibrary[arrows,shapes,snakes]
\usetikzlibrary[patterns]
\usepackage{verbatim}
\usepackage{amsthm}
\usepackage{appendix}
\usepackage{authblk}

\newtheorem{theorem}{Theorem}

\newtheorem{axiom}{Axiom}

\newtheorem{definition}[axiom]{Definition}

\newtheorem{lemma}[theorem]{Lemma}

\renewcommand{\eqref}[1]{Eq.~(\ref{#1})}
\begin{document}

\title{\huge{DO WE UNDERSTAND QUANTUM MECHANICS  -  FINALLY?}}
\normalsize
\author[*]{J\" urg Fr\" ohlich}
\author[**]{Baptiste Schubnel}
\affil[*]{\textit{\small{Institute for Theoretical Physics, ETH Zurich, CH-8093, Zurich, Switzerland}}}
\affil[**]{\textit{\small{Department of Mathematics, ETH Zurich, CH-8092, Zurich, Switzerland}}}	

\date{February 2012}

\maketitle 
\small
\tableofcontents
\normalsize

\bigskip
\textit{``If someone tells you they understand quantum mechanics then all you've learned is that you've met a liar.''} (R.P. Feynman)\\
\textit{``Anyone who is not shocked by quantum theory has not understood it.''} (N. Bohr)

In these notes we present a short and necessarily rather rudimentary summary of some of our understanding of what kind of a physical theory of Nature Quantum Mechanics is. They have grown out of a lecture the senior author presented at the Schr\"odinger memorial in Vienna, in January 2011. A more detailed and more pedagogical account of our view of quantum mechanics, attempting to close various gaps in the mathematics and physics of these notes, will be published elsewhere. We do not claim to offer any genuinely new or original thoughts on quantum mechanics. However, we have made the experience that there is still a fair amount of confusion about the deeper meaning of this theory -- even among professional physicists. The intention behind these notes (and a more detailed version thereof) is to make a modest contribution towards alleviating some of this confusion.\\
After a short introductory section on the history of Schr\"odinger's wave mechanics, we will sketch a unified view of non-relativistic theories of physical systems comprising both classical and quantum theories. This will enable us to highlight the fundamental conceptual differences between these two classes of theories. Our goal is to sketch what it is that quantum mechanics predicts about the behavior of physical systems when appropriate experiments are made, and in what way it differs radically from classical theories. Incidentally, we hope to convince the reader that Bohr and Feynman may have been a little too pessimistic in their assessment of our understanding of  this wonderful theory.\\
Our main results may be found in Sects. 2, 3.1, 3.2 and 4.

\textit{Acknowledgements.}\\
The senior author has learnt most of what he understands about Quantum Mechanics from Markus Fierz, Klaus Hepp and Res Jost, many years ago when he was a student at ETH Zurich. This did not spare him sufferings through prolonged periods of confusion about the nature of the theory, later on. Apparently, there is no way around thinking about these things and trying to clarify one's thoughts, all by oneself. He is especially grateful to Klaus Hepp and Norbert Straumann for plenty of hints that led him towards some understanding of various elements of the foundations of the theory.
He acknowledges many useful discussions with Peter Pickl and Christian Schilling. He thanks his friends in Vienna for having invited him to present a lecture on Quantum Mechanics at the Schr\"odinger memorial and Jakob Yngvason for insightful comments.\\
 These notes were completed during a stay at the ``Zentrum f\"ur interdisziplin\"are Forschung'' (ZiF) of the University of Bielefeld. The authors thank their colleague and friend Philippe Blanchard and the staff of the ZiF for their very friendly hospitality.\\
\textit{These notes are dedicated to the memory of Ernst Specker (1920-2011), who set a brilliant example of a highly original and inspiring scientist and teacher for all those who had the privilege to have known him. His contributions to Quantum Mechanics will be remembered.}\\

\section{Schr\"{o}dinger and Zurich}
With Heisenberg and Dirac, \textit{Schr\"odinger} is one of the fathers of (non-relativistic) quantum mechanics in its final form. Everyone has heard about his wave mechanics and about the Schr\"odinger Equation.
He made his most important discoveries during a time when he held a professorship at the University of Zurich. With Berlin, Berne, G\"ottingen and Cambridge (UK), Zurich was one of the birth places of the new theories of 20th Century Physics. It may thus be appropriate to begin with a short summary of some important facts about ``Schr\"odinger and Zurich''. The sources underlying our summary are \cite{Uhlmann}, \cite{Thirring}.\\
\textit{Erwin (Rudolf Josef Alexander) Schr\"odinger} was born in Vienna, on August 12, 1887. His father was catholic, his mother a Lutheran. Her mother's mother was English. German and English were spoken at home. Erwin started to study Mathematics and Physics at the University of Vienna in 1906. After only four years, he was promoted to \textit{Dr. phil.}, in 1910. Among his teachers were \textit{Franz S. Exner} and $Friedrich$  $Hasen$\"o$hrl$. The latter was killed in 1915, during World War I. Through Hasen\"{o}hrl's influence, 
Schr\"odinger liked to think of himself as a student of \textit{Ludwig Boltzmann}. He recognized Hasen\"ohrl's scientific importance and influence on his own work in his Nobel lecture. In \textit{``Mein Leben, meine Weltsicht''}, Schr\"odinger writes: \textit{``Ich m\"ochte nicht den Eindruck hinterlassen, mich h\"atte nur die Wissenschaft interessiert. Tats\"achlich war es mein fr\"uher Wunsch, Poet zu sein. Aber ich bemerkte bald, dass Poesie kein Geld einbringt. Die Wissenschaft dagegen offerierte mir eine Karriere.''}\\
Glancing through Schr\"odinger's early work, one notices his talent for language and his pragmatism in choosing seemingly promising research topics. One also encounters many signs of his excellent mathematical education and his talent for mathematical reasoning. No wonder science offered him a rather smooth career. However, before his appointment as ``Ordinarius f\"ur Theoretische Physik'' at the University of Zurich, in the fall of 1921, to the chair previously held by $Einstein$ and $von$ $Laue$, there were only few signs of his extraordinary genius. Schr\"odinger's years in Zurich constitute, undoubtedly, the most creative period 
in his life.\\
His first important paper, which concerned an application to quantum theory of \textit{Weyl's} idea of the electromagnetic field as a gauge field, was submitted for publication in 1922. His epochal papers on wave mechanics only followed a little more than three years later. The first one, ``Quantisierung als Eigenwertproblem (Erste Mitteilung)'', was submitted for publication on January 27, 1926; the second one (same title - Zweite Mitteilung) on February 23, the third one (classical limit of wave mechanics) shortly thereafter, the fourth one (equivalence of wave- and matrix mechanics) on March 18, the fifth one (Dritte Mitteilung - perturbation theory and applications) on May 10, the sixth one (Vierte Mitteilung - time-dependent Schr\"odinger equation, time-dependent perturbation theory) on June 21, the seventh one -- a summary of his new wave mechanics published in the Physical Review -- on September 3, the eighth one (Compton effect) on December 10 -- all during the year of 1926. In November 1926, he completes his ``Vorwort zur ersten Auflage'' of his \textit{``Abhandlungen zur Wellenmechanik''}.\\
The intensity of Schr\"odinger's scientific creativity and productivity, during that one year, may well be without parallel in the history of the Natural Sciences, with the possible exception of $Einstein's$ ``annus mirabilis'', 1905. Schr\"odinger discovers all the right equations, all the right concepts and all the right mathematical formalism. Mathematics-wise, he is well ahead of  his competitors, except for $Wolfgang$ $Pauli$. He talks about linear operators, introduces Hilbert space into his theory, addresses and solves many of the pressing concrete problems of the new quantum mechanics -- and somehow misses its basic message. He is haunted by the philosophical prejudice that physical theory has to provide a realistic description of Nature that talks about what happens, rather than merely about what $might$ happen. His goal is to find a description of phenomena in the microcosmos in the form of a classical relativistic wave-field theory somewhat analogous to Maxwell's theory of the electromagnetic field -- of course without succeeding. In spite of his philosophical prejudices, he is to unravel some of the most important concepts typical of the new theory, such as entanglement and decoherence, and to arrive at all the right conclusions -- apparently without ever feeling comfortable with his own discoveries.\\
On  the first of October of 1927, Schr\"odinger was appointed as the successor of $Max$ $Planck$ in Berlin. In 1933, he shares the Nobel Prize in Physics with $Paul$ $(Adrien$ $Maurice)$ $Dirac$. In 1935, in the middle of the turmoil around his emigration from Nazi Germany, he conceives his famous ``Schr\"odinger's cat'' Gedanken experiment, which introduces the idea of decoherence.\\
Coming from Zurich, we feel we should ask why this city was the right place where Schr\"odinger could make his epochal discoveries. Berne and Zurich were the cities where Einstein had made his most essential discoveries in quantum theory. Thanks to the presence of many famous refugees of World War I, Zurich was a rather cosmopolitan city with a liberal spirit. The scientific atmosphere created by Einstein, von Laue, $Debye$, Weyl and others must have been fertile for discoveries in quantum theory. In his work on wave mechanics, Schr\"odinger was, according to his own testimony, much influenced by Einstein's work on ideal Bose gases (1924/25) and \textit{de Broglie's} work on matter waves. It is reported that Debye directed Schr\"odinger towards de Broglie's work and suggested to him to look for a wave equation describing matter waves. In Arosa, where Schr\"odinger repeatedly spent time to cure himself from tuberculosis, he apparently discovered the right equation. Rumor has it that he pursued his ideas at the ``Dolder Wellenbad'', a swimming pool above Zurich with artificial waves (and pretty women sun bathing on the lawn).\\
It is appropriate to mention the role Weyl played as Schr\"odinger's mathematical mentor, during their Zurich years. It was Weyl who apparently explained to Schr\"odinger that his time-independent wave equation represented an eigenvalue problem and directed him to the right mathematical literature. Schr\"odinger acknowledges this in his first paper on wave mechanics. Apart from his superb knowledge of mathematics, Weyl was intensely familiar with modern theoretical physics including quantum theory. He had prescient ideas on some of the radical implications of quantum theory (such as its intrinsically statistical nature and problems surrounding the notion of an ``event'' in the quantum world), quite some time before matrix- and wave mechanics were discovered. Weyl was Schr\"odinger's senior by only two years. They were close friends. It is thus plausible that Weyl played a rather important role in the development of Schr\"odinger's thinking. Their relationship is a model for the fruitfulness of interactions between mathematicians and theoretical physicists.\\
While the physical arguments that led $Heisenberg$ and, following his lead, Dirac to the discovery of quantum mechanics (in the form of matrix mechanics and transformation theory) appear to us as relevant and fresh as ever, Schr\"odinger's formal arguments based on an analogy with optics,
\begin{equation}
\nonumber
\text{geometrical optics}:\text{wave optics} \sim \text{Hamiltonian mechanics}:\text{wave mechanics},
\end{equation}
may nowadays seem to be of mainly historical interest. Although they initially misled him to an erroneous interpretation of wave mechanics, they made him discover very powerful mathematical methods from the theory of partial differential equations and of eigenvalue problems that his competitors did not immediately recognize behind their more abstract formulation of the theory.

But it is time to leave science history and proceed to somewhat more technical matters.

\section{What is a physical system, mathematically speaking?}

In this section, we outline a mathematical formalism suitable for a unified description of classical and quantum-mechanical theories of physical systems.  It is most conveniently formulated in the language of operator algebras; see, e.g., \cite{TA} or \cite{BR}.
We suppose that there is an observer, $\mathcal{O}$, who studies a physical system, $S$. To gather information on $S$,  $\mathcal{O}$ performs series of experiments designed to measure various physical quantities pertaining to $S$, such as  positions, momenta or spins of some particles belonging to $S$. 
No matter whether we speak of classical or quantum-mechanical theories of physical systems, physical quantities are always represented, mathematically, by (bounded) linear operators. For classical systems, they correspond to real-valued functions on a space of pure states (phase space, in the case of Hamiltonian systems) acting as multiplication operators on a space of half-densities over the space of pure states; for quantum-mechanical systems with finitely many degrees of freedom, they correspond to (non-commuting) selfadjoint linear operators acting on a separable Hilbert space. \\
In these notes, we study non-relativistic theories of physical systems, i.e., we assume that signals can be transmitted arbitrarily fast. It is then reasonable, in either case, to imagine that the physical quantities pertaining to $S$ generate some $^{*}$-algebra of operators,  denoted by $\mathcal{A}_{S}$, that does not depend on the observer $\mathcal{O}$. (In contrast, in general relativistic theories, the algebra of physical quantities not only depends on the choice of a physical system but will also depend, in general, on the observer.) \\ 
Given a physical quantity represented by a selfadjoint operator $a \in \mathcal{A}_{S}$, a measurement of $a$ results in an $``event"$ \cite{Haag} corresponding to some measured value of $a$. Since measurements have only a finite precision, one may associate with every such event a real interval, $I$, describing a range of possible outcomes in a particular measurement of $a$. It is natural to associate to this possible event the corresponding \textit{spectral projection}, $P_{a}(I)$, associated with the selfadjoint operator $a$ and the interval $I\subset\mathbb{R}$ via the spectral theorem. 
Thus, spectral projections associated with selfadjoint operators in $\mathcal{A}_{S}$ corresponding to physical quantities of $S$ represent \textit{possible events} in $S$. 
\\Generally speaking, one may argue that there are events happening in $S$ that are not necessarily triggered by an actual measurement (undertaken by an observer using experimental equipment) of some physical quantity represented by an operator in the algebra $\mathcal{A}_{S}$ but rather by interactions of $S$ with its environment. It is plausible to assume that $any$ possible event in $S$ can be represented by an orthogonal projection, $P$, (or, more generally, by a positive operator-valued measure). The operator $P$ is a mathematical representation of the acquisition of information about $S$; it does not represent a physical process. It is assumed that all possible events in $S$ generate a $C^{*}$- or a von Neumann algebra, henceforth denoted by $\mathcal{B}_{S}$. For simplicity, it will always be assumed that the algebra $\mathcal{B}_{S}$ contains an identity operator. The algebra $\mathcal{A}_{S}$ is contained in or equal to the algebra $\mathcal{B}_{S}$. (In these notes, we will be somewhat sloppy about the right choices of these algebras. At various places, this will undermine their mathematical precision. Things will be rectified in a forthcoming essay.) \\
$States$ of the system $S$ are identified with states on the algebra  $\mathcal{B}_{S}$, i.e., with normalized, positive linear functionals on $\mathcal{B}_{S}$.\\ 
The purpose of a theory of a physical system $S$ is to enable theorists to predict the probabilities of (time-ordered) sequences of possible events in $S$ -- $``histories"$ of $S$ -- to actually happen when $S$ is coupled to another system, $E$, needed to carry out appropriate experiments, given that they know the state of the system corresponding to the composition of $S$ with $E$. We emphasize that $E$ is treated as a physical system, too, and that it plays an important role in associating ``facts'' with ``possible events'' in $S$; (this being related to the mechanisms of ``dephasing'' and ``decoherence''). Generally speaking, $E$ can either correspond to some \textit{experimental equipment} used to observe $S$, or to some \textit{environment} $S$ is coupled to. We think of $E$ as ``experimental equipment'' if the initial state of $E$ and its dynamics can be assumed to be controlled, to some extent, by  an observer $\mathcal{O}$, (an experimentalist who can turn various knobs and tune various parameters). If the state and the dynamics of $E$ are beyond the control of any observer we think of $E$ as  ``environment''. Of course, the distinction is usually not sharp. It is important to understand why probabilities of histories of $S$ do not sensitively depend on precise knowledge of the state and the dynamics of $E$. (See Sect. 3.3 for a result going in this direction. We plan to return to these matters elsewhere.)\\
A $``realistic"$ (or $``deterministic"$) theory of a physical system $S$ is characterized by the properties that any possible event in $S$ has a complement, in the sense that $either$ the event $or$ its complement $will$ happen, and that if the state of $S\vee{E}$ is $pure$ the probability of a possible history of $S$ is either $=0$ (meaning that it will never be observed) or $=1$ (meaning that it will be observed with certainty). We say that pure states of a system described by a realistic theory give rise to \textit{``0--1 laws''} for the probabilities of its histories. 
In contrast, a \textit{``quantum theory''} is characterized by the properties that, in general, there may be ``interferences'' between a possible event and its complement -- meaning that they do  not mutually exclude each other -- and that there are pure states that predict strictly positive probabilities that are strictly smaller than 1 for certain histories. A quantum theory is therefore intrinsically $non{-}deterministic$.
These remarks make it clear that a crucial point to be clarified is how one can $prepare$ a physical system in a specific state (pure or mixed, in case the state is only partly specified) of interest in an experiment that theorists want to make predictions on. This point has been studied, in a rather satisfactory way, for a respectable class of quantum theories. However, the relevant results and the methods used to derive them go beyond the scope of this review. They will be treated in a forthcoming paper; (but see, e.g., \cite{FGrS}, \cite{DeRoeck}).\\
We hope that the meaning of the notions and remarks just presented will become clear in the following discussion.\\
John von Neumann initiated the creation of the theory of operator algebras, in order to have a convenient and precise mathematical language to think and talk about quantum physics and to clarify various mathematical aspects of the theory. There is simply no reason not to profit from his creation -- no apologies!
Thus, as announced, above, we will consider the rather vast class of theories of physical systems that can be formulated in the language of operator algebras. Such theories are further characterized by specifying the following data.

\begin{definition}

Mathematical data characterizing a theory of a physical system $S$
\end{definition}

\begin{enumerate}[(I)]
 \item A $C^{*}$-algebra, $\mathcal{B}_{S}$, generated by ``all'' possible events in $S$, containing the $^{*}$-algebra $\mathcal{A}_{S} \subseteq \mathcal{B}_{S}$ generated by physical quantities pertaining to $S$.
\item The  convex set of states, $\mathcal{S}_{S}$, on the algebra $\mathcal{B}_{S}$.
\item A group of $symmetries$, $\mathcal{G}_{S}$, of $S$, including time evolution. Elements $g \in \mathcal{G}_{S}$ are assumed to act as $^{*}$-automorphisms, $\alpha_{g}$, on $\mathcal{B}_{S}$. The group of all $^{*}$-automorphisms of $\mathcal{B}_{S}$ is denoted by $Aut(\mathcal{B}_{S})$.\\
We remark that the algebra $\mathcal{B}_{S}$ and the group of symmetries $\mathcal{G}_{S}$ depend on the environment $S$ is coupled to.
\item \textit{Subsystems}: $S$ is a subsystem of $S'$, $S \subset S'$, iff $\mathcal{B}_{S} \subset \mathcal{B}_{S'}$.\\ 
\textit{Composition of systems}: If $S$,$S'$ are two systems and $\bar{S}=S \vee S'$ denotes their composition then $\mathcal{B}_{\bar{S}} \equiv \mathcal{B}_{S \vee S' }=\mathcal{B}_{S} \otimes \mathcal{B}_{S'}$. 
\\If $S \simeq S'$ then
one must specify an embedding of the state space $\mathcal{S}_{S\vee S} \subseteq \mathcal{S}_{S}  \otimes \mathcal{S}_{S}$. This is the issue of $statistics$, which plays a crucial role in quantum mechanics; (Fermi-Dirac-, Bose-Einstein-, or fractional statistics).

\end{enumerate}

 The choice of (I) and (III) depends on the experimental equipment available to observers exploring the system $S$. To illustrate this point, think of the solar system, $S_{\odot}$. An astronomer in the times of Tycho Brahe would have chosen much fewer physical quantities to describe possible events in $S_{\odot}$ than a modern astrophysicist equipped with the latest instruments. This results in drastically different choices of  algebras $\mathcal{A}_{S_{\odot}}$ and $\mathcal{B}_{S_{\odot}}$, in spite of the fact that the actual physical system remains the same; (Tycho Brahe would obviously have chosen much smaller such algebras than a contemporary astrophysicist, and, thus, their theoretical descriptions of the solar system would drastically differ from one another). Well, our theories of physical systems are but images of such systems inside some mathematical structure. These images are never given by ``isomorphisms''; they are more or less coarse-grained (depending on the experimental equipment  and the precision of the data available to us), and our choice of mathematical structures as screens for images of physical systems need not be unique.
 
It is remarkable that  the new  physical theories of the 20th Century appear to arise from older (precursor) theories by $``deformation$'' of the structures (I), (III) and (IV).  For example, quantum mechanics can be obtained from classical Hamiltonian mechanics by deforming the algebra $\mathcal{A}_{S}$ from a commutative, associative algebra to a non-commutative, associative one - $``quantization"$ - (theory of deformations of associative algebras), the deformation parameter corresponding to Planck's constant $\hbar$. One can view atomistic theories of matter as arising from (Hamiltonian) theories of continuous media by a deformation of the algebra $\mathcal{A}_{S}$, the deformation parameter corresponding (roughly speaking) to the inverse of Avogadro's number $N_A$. By deforming the Galilei symmetry of non-relativistic systems one is led to the Poincar\'e symmetry of (special) relativistic systems; the deformation parameter is the inverse of the speed of light $c$. This is an example of a deformation of (III) (deformation of Lie groups and -algebras) leading to new physical theories. Fractional statistics -- a form of quantum statistics encountered in certain two-dimensional systems, in particular, in 2D electron gases exhibiting the fractional quantum Hall effect -- which was overlooked by the pioneers of quantum theory, can arise as a deformation of ordinary Bose-Einstein or Fermi-Dirac statistics (deformation of braided tensor categories).\\
The ``deformation point of view'' alluded to here was originally proposed by Moshe Flato \cite{Flato} and taken up by Ludwig Faddeev. Some elements of it are sketched in Section 3.4; (see also \cite{FRKP}, \cite{FRKP'} and refs. given there).

\section{Realistic theories versus quantum theories}
In this section, we introduce two distinct classes of physical theories. A physical theory is called $\it``realistic$'' if the algebra $\mathcal{B}_{S}$  is abelian (commutative). It is called $\it``quantum$'' if $\mathcal{B}_{S}$ is non-abelian (non-commutative). We will see that there is an intimate connection between the commutativity of $\mathcal{B}_{S}$  and determinism - determinism necessarily fails if  $\mathcal{B}_{S}$ is non-commutative.

\subsection{Realistic theories}

In this section, we summarize some of the most important features of realistic (or deterministic) theories.

\subsubsection{Characterisation of $\mathcal{B}_{S}$ and  $\mathcal{S}_{S}$ in realistic theories}

\textit{Realistic theories} of physical systems are theories with an $abelian$ algebra $\mathcal{B}_{S}$ of possible events. Important examples of realistic theories are Hamiltonian systems. For such systems, the algebra of possible events $\mathcal{B}_{S}$ is given by the algebra of bounded (continuous or measurable) functions on the phase space, $\Gamma_{S}$, (some symplectic manifold) of the system $S$ composed with the environment it is interacting with, and the algebra $\mathcal{A}_{S}$ is some subalgebra contained in or equal to $\mathcal{B}_S$. Phase space $\Gamma_{S}$ is equipped with a symplectic form, $\sigma_S$, (a closed, non-degenerate 2-form on $\Gamma_{S}$), which gives rise to a Poisson bracket, $\lbrace f,g \rbrace=\sigma_S (X_f,X_g)$, on $\mathcal{B}_{S}$, with $X_f$ denoting the Hamiltonian vector field corresponding to the function $f \in \mathcal{B}_{S}$. This furnishes $\mathcal{B}_{S}$ with the structure of a Lie algebra. States on $\mathcal{B}_{S}$ are given by probability measures on $\Gamma_{S}$, pure states are given by Dirac measures ($\delta$ - functions). 

In this section, we wish to consider general realistic theories. Let $\mathcal{B}_{S}$ be the abelian $C^{*}$-algebra of possible events of a  realistic theory of a physical system $S$.  We denote the set of non-zero homomorphisms of $\mathcal{B}_{S}$ into $\mathbb{C}$ by $M_S$, called the $spectrum$ of $\mathcal{B}_{S}$.  One can prove that $M_S$ is locally compact  in the $\sigma(\mathcal{B}_{S}^{*}, \mathcal{B}_{S})$- topology, and that it is a compact Hausdorff space if $\mathcal{B}_{S}$ is unital, i.e., contains an identity, $\mathbb{I}$. It is appropriate to recall a famous theorem due to I.M. Gel'fand.
\begin{theorem}
If $\mathcal{B}_{S}$ is an abelian $C^{*}$-algebra then it  is isometrically isomorphic to the abelian  $C^{*}$-algebra, $\mathcal{C}_{0}(M_S)$, of all continuous functions on $M_S$ vanishing at infinity.
\end{theorem}

This isomorphism is given by the Gel'fand transform, $\hat{(\cdot)}$, that assigns to each $b \in \mathcal{B}_{S}$, the function $\hat{b}$ acting on $M_S$ by $\hat{b}(\omega)=\omega(b)$, for all $\omega \in M_S$. To say that $b$ ``vanishes at infinity'' means that, for all $\epsilon>0$, the set $\lbrace \omega \mid  \hat{b}(\omega) \geq \epsilon \rbrace$ is compact in the $\sigma(\mathcal{B}_{S}^{*}, \mathcal{B}_{S})$- topology. We remark that the properties of $M_S$ crucially depend on the precise choice of the algebra $\mathcal{B}_{S}$, and this fact would require more attention than it is given in these notes.

The set of states of realistic theories can be characterized using well known results from measure theory. In particular, the following theorem due to Riesz and Markov is relevant. (We will assume that the algebra $\mathcal{B}_{S}$ is unital, hence $M_S$ is compact.)
\begin{theorem}
Let M be a compact measure space. Then every positive linear functional, $\omega$, on $\mathcal{C}_{0}(M)$ is given by a unique Baire measure, $\mu_{\omega}$, on  $M$, with $\omega(f)=\int_{M} f d\mu_{\omega}$.
\end{theorem}
$Remark.$ The measure $\mu_{\omega}$ can be uniquely extended to a regular finite Borel measure. If we restrict physical states to normalized states (norm unity) the corresponding Borel measure is a \textit{probability  measure}, because $\vert \vert \omega \vert \vert =\omega(1)=1$. 

Let us look at the family of pure states of a realistic theory. We recall  that a pure state is an extremal element of the convex set of states $\mathcal{S}_{S}$, i.e., $\omega$ is pure iff it cannot  be written in the form $\omega=\lambda \omega_{1} + (1-\lambda) \omega_{2}$, where $0< \lambda <1$, and $\omega_{1} \neq \omega_{2} $ belong to $\mathcal{S}_{S}$. On an abelian  $C^{*}$-algebra, every pure state $\omega$ is multiplicative. Consequently,  
$$0=\omega(f- \omega(f)) \omega(f^*- \omega(f^*))=\omega(\mid f-\omega(f) \mid^{2})=\int_{M_S}  \mid f-\omega(f) \mid^{2} d\mu_{\omega}$$
i.e.,  every $f \in \mathcal{C}_{0}(M_S)$ is $\mu_{\omega}$- almost everywhere constant, which implies that $\mu_{\omega}=\delta_{x}$, for some point $x \in M_{S}$. Pure states are thus Dirac $\delta$- measures on $M_{S}$, and the map $x \rightarrow \delta_{x}$ is a homeomorphism, because $M_{S}$ is completely regular. It follows immediately that there is no \textit{superposition principle} within the set of pure states of realistic theories, because any linear combination of the half-densities corresponding to two distinct pure states (given by $\delta$- functions with disjoint supports) is the half-density of a $mixed$ state. 

We have thus identified some typical features of realistic theories: $\mathcal{B}_{S}$ is of the form $\mathcal{B}_{S}=\mathcal{C}_{0}(M_S)$, where $M_S$ is a (locally) compact Hausdorff space; pure states are given by Dirac $\delta$- functions, and  general states are given by probability measures on  $M_S$. 

\subsubsection{Composition of systems}
We consider two systems, $S_{1}$ and $S_{2}$, and propose to clarify what is meant by their composition $S_{1}\vee S_{2}$. The abelian algebra of  possible events of the composed system is $\mathcal{B}_{S_{1} \vee  S_{2} }= \mathcal{C}_{0}(M_{S_1} \times M_{S_2}  )$. We denote by $T_{1}$ ($T_{2}$) the $\sigma$-algebra of Borel sets on   $M_{S_1}$ ($ M_{S_2}$) with respect to the $\sigma(\mathcal{B}_{S_1}^{*}, \mathcal{B}_{S_1})$-  ($\sigma(\mathcal{B}_{S_2}^{*}, \mathcal{B}_{S_2})$-) topology. If $M_{S_1} \times M_{S_2}$ is equipped with the $\sigma$-algebra $T_{1} \times T_{2}$,  then probability measures $\mu_1$ on $ M_{S_1}$ and $\mu_2$ on $ M_{S_2}$ define a state of the composed system given by the tensor product measure $\mu_1 \otimes \mu_2$.  Every pure state of the composed system is a Dirac $\delta$- measure on the product space $M_{S_1} \times M_{S_2} $,  of the form $\delta_{(x_1,x_{2})}=\delta_{x_1} \otimes \delta_{x_2} $. Thus, every pure state of the composed system is ``separable'', i.e., remains pure when restricted to a subsystem, and hence there is no interesting notion of entanglement between $S_{1}$ and $S_{2}$.

\subsubsection{Symmetries of $\mathcal{B}_S$}

We recall that symmetries of $S$ are represented by $^{*}$-automorphisms of $\mathcal{B}_{S}$. If this algebra is abelian then there is a one-to-one correspondence between $^{*}$-automorphisms of $\mathcal{B}_S=\mathcal{C}_{0}(M_{S})$ and homeomorphisms of $M_{S}$.  Indeed, let $\alpha$ be a $^{*}$-automorphism of $\mathcal{C}_{0}(M_{S})$ and $\omega$ be a state on $\mathcal{C}_{0}(M_{S})$. Then $\hat{\alpha}(\omega)=\omega \circ \alpha$ is again a state on $\mathcal{C}_{0}(M_{S})$. It is multiplicative if $\omega$ is multiplicative. Moreover, it is clear that $\hat{\alpha}:\mathcal{S}_{S} \rightarrow  \mathcal{S}_{S}$ is a bijection, with inverse $\hat{\alpha}^{-1}(\omega)=\omega \circ \alpha^{-1}$.  The map $\hat{\alpha}$ from $M_{S}$ to itself is then also a bijection. If $\delta_{x_n}$ converges to $\delta_{x}$ in the $\sigma(\mathcal{B}_{S}^{*}, \mathcal{B}_{S})$- topology, then $(\hat{\alpha}( \delta_{x_n})-\hat{\alpha}( \delta_{x}))(f)=\alpha(f)(x_n)-\alpha(f)(x)$ converges to $0$, too, for any $f \in \mathcal{C}_{0}(M_{S})$. Thus, $\hat{\alpha}$ is a homeomorphism from $M_{S}$ to $M_{S}$.  The other direction is obvious.

In our effort to rediscover typical features of classical dynamical systems within the general algebraic formalism developed here, it is natural to ask the following questions:
\begin{enumerate}
\item Under which assumptions on the algebra $\mathcal{B}_{S}$ does the spectrum $M_S$ admit a tangent bundle; (in particular, when is $M_S$ a manifold)? What are the smoothness properties of $M_{S}$?
\item Under which hypotheses does $M_S$ have the structure of a symplectic manifold? 
\end{enumerate}
Some useful references for question (1) are \cite{WE, BRa}. It has been studied in depth in a recent paper of Connes, \cite{CO}. Connes uses methods of his non-commutative geometry to prove that a spectral triple, $(\mathcal{A}, \mathcal{H},D)$, where $\mathcal{A}$ is a commutative $^{*} \text{-algebra}$ of bounded linear operators acting on a Hilbert space $\mathcal{H}$, $D$ is a self-adjoint operator acting on $\mathcal{H}$ whose commutator with any element of $\mathcal{A}$ is an operator commuting with the elements of $\mathcal{A}$, fullfilling certain rather subtle additional requirements, has the property that the algebra $\mathcal{A}\simeq\mathcal{C}^{\infty}(M)$, where $M$ is a smooth compact manifold. The operator $D$ is a generalization of the Dirac operator on a spin manifold. In our context, it would be preferable to formulate conditions on the Lie algebra of derivations, $\mathcal{D}_{S}$, of some $^{*}$-subalgebra, $\mathcal{B}_ {S}^{0}$, weakly dense in the algebra $\mathcal{B}_S$ that guarantee that $M_S$, now defined as the spectrum of $\mathcal{B}_{S}^{0}$, admits a tangent bundle, whose sections can be identified with the elements of $\mathcal{D}_S$.

The second question has been considered, within a general algebraic formalism, in \cite{FGR}. But there do not appear to exist satisfactory general answers, yet.

Why are we interested in the first question stated above? Well, if $M_{S}$ has a tangent bundle we have a natural notion of vector fields. We may then study a one-parameter family  $\alpha:I \rightarrow \text{Aut}(\mathcal{B}_S)$ of $^{*}$-automorphisms, where $I$ is an interval of $\mathbb{R}$ containing $0$, such that $\alpha_0=id$. According to the previous discussion, this family gives rise to a one-parameter family of homeomorphisms $\hat{\alpha}:I \rightarrow \text{Homeo}(M_S) $, with $\hat{\alpha}_{0}=id$. If $M_S$ has a tangent bundle, $TM_{S}$, the map $X:I \times M_S \rightarrow TM_{S}$, formally given by 
$$X(t,x)=\frac{d}{ds} \left( \hat{\alpha}_{s} (\hat{\alpha}_{t}^{-1}(x))\right)_{s=t},$$  
may be a well-defined, smooth time-dependent vector field, with $$X(t,\hat{\alpha}_{t}(x))=\frac{d}{ds} \left( \hat{\alpha}_{s} (x) \right)_{s=t},$$
meaning that $t \mapsto \hat{\alpha}_{t}(x)$ is an integral curve of $X$ that passes through $x$ at time $t=0$. One can check that if $\hat{\alpha}_{t}$ satisfies $\hat{\alpha}_{t+s}=\hat{\alpha}_{t} \circ \hat{\alpha}_{s}$, then the vector field $X$ is time-independent, because
$$X(x)=\frac{d}{ds} \left( \hat{\alpha}_{s}(x))\right)_{s=0}, \text{independent of $t$}.$$
Conversely, using  the Cauchy-Lipschitz theorem, one shows that every time-dependent Lipschitz-continuous vector field $X: I\times M_S \rightarrow M_S$ ($M_S$ assumed to be compact), where $I$ is an open interval of $\mathbb{R}$ containing 0, generates a unique globally defined flow, $\hat{\alpha}$, with initial conditions $\hat{\alpha}_{0}(x)=x$. For each $t \in I$, $\hat{\alpha}_{t}$ is a homeomorphism of $M_S$. If the vector field $X$ is time-independent then $\hat{\alpha}_{t+s}=\hat{\alpha}_{t} \circ \hat{\alpha}_{s}$, for $t,s,t+s \in I$. If $X$ is time-dependent this relation does not hold, because, for $s$ fixed, $\hat{\alpha}_{t+s}(x)$ is not an integral curve of $X$, but of $Y(t,\cdot)=X(t+s,\cdot)$. We can label each integral curve with an extra index indicating the initial time, using the notation $\hat{\alpha}_{t,s}$ for the flow maps. One then shows that $\hat{\alpha}_{t,s} \circ \hat{\alpha}_{s,u} = \hat{\alpha}_{t,u} $, for arbitrary $t,s,u \in I$.

A realistic physical theory, i.e., one with an abelian algebra $\mathcal{B}_{S}$, for which  a family, $\alpha_{t,s}$, of \mbox{$^{*}$-automorphisms} describing time-translations of elements in $\mathcal{B}_S$ is specified, is an example of a  $\textit{deterministic}$ dynamical system. 
As explained above, the $^{*}$-automorphisms $\alpha_{t,s}$ determine a family of homeomorphisms, $\hat{\alpha}_{t,s}$, of $M_S$ generated by a time-dependent vector field. 
We consider a family of Borel sets, $\Omega_i$,  of $M_S$ and possible events
$$P_i:=\alpha_{t_i,t_0} (\chi_{\Omega_i})=\chi_{\Omega_i} \circ \hat{\alpha}_{t_i,t_0}=\chi_{\hat{\alpha}_{t_0,t_i}(\Omega_i)}, \text{ }i = 1,...,n,$$
where $\chi_{\Omega}$ is the characteristic function of $\Omega$.  The possible event $P_i$ corresponds to a $fact$, given an initial state $x\in M_S$ at time $t_0$, if and only if the state $x_{t_i}=\hat{\alpha}_{t_i,t_0}(x)$ of the system at time $t_i$ belongs to the set $\Omega_i$. Note that the maps $\hat{\alpha}_{t,t_0}$ are uniquely determined by the initial condition $\hat{\alpha}_{t_0,t_0}(x)=x$ and the vector field $X(t_0+(\cdot),\cdot)$. Because the integral curves are continuous in time, the probability of observing the sequence of events $P_{1},...,P_{n}$, represented by the function $\prod_{i} P_i$, in a pure state $\delta_{x}$, for some point $x \in M_S$, is given by $\delta_{x} \left(\prod_{i} P_i \right)$ and is either ${0}$ or ${1}$; i.e., pure states always give rise to $\textit{ ``0--1 laws''}$. The failure of the probability of a sequence of events in a state $\omega$ of the system to satisfy a ``0--1 law'' implies that $\omega$ is mixed.  These are features that characterize a deterministic dynamical system.

It is sometimes useful to generalize our notion of time-evolution of a realistic physical system by defining it in terms of one-parameter families of maps from the space of states, $\mathcal{S}_{S}$, of the system to itself (``Schr\"odinger picture''), not necessarily requiring that pure states are mapped to pure states, rather than as $^{*}$-automorphisms of $\mathcal{B}_{S}$ (``Heisenberg picture''). This immediately leads one to the theory of stochastic processes over the state space $M_S$.

\subsection{Quantum theories}

The unique and only feature that distinguishes a $\textit{quantum theory}$ from a realistic theory is that, in a quantum theory, $\mathcal{A}_S$, and thus $\mathcal{B}_S$, are $non{-}abelian$ operator algebras. This entails that quantum theories show many features not encountered in realistic theories. The best known example are Heisenberg's uncertainty relations, which are an immediate consequence of the non-commutativity of $\mathcal{B}_S$. The main new feature exhibited by quantum theories is, however, that they are intrinsically non-deterministic. This fact has caused plenty of grief and confusion among physicists. It is a fairly direct consequence of the non-commutativity of $\mathcal{B}_{S}$, which implies that the probabilities of sequences of possible events (i.e., histories) do, in general, $not$ follow 0--1 laws, anymore, \textit{even if the state of $S \vee E$ ($E$ the equipment used to explore $S$) used to predict such probabilities is pure}.

\subsubsection{Uncertainty relations}
Let us consider selfadjoint operators $a,b \in \mathcal{B}_S$.  For any state $\omega$ and an arbitrary $\mu \in \mathbb{C}$, $\omega((a+\bar{\mu} b)(a+\mu b))\geq 0$, which, by a standard argument, implies that $\omega(a^{2}) \omega(b^{2}) \geq \frac{1}{4} \vert \omega(\left[ a,b\right])\vert^{2} $, where \mbox{$\left[ a,b\right]=ab-ba$}. Assuming, without loss of generality, that $\omega(a)=\omega(b)=0$, one recovers the standard Heisenberg uncertainty relations. (The original form concerns the case where $\left[ a,b\right]=i\lambda \mathbb{I}$, $\lambda \in \mathbb{R}$. -- For time-energy uncertainty relations, see, e.g., \cite{PfF}.)

\subsubsection{Hilbert space formalism and superposition principle for pure states}
The GNS (Gel'fand-Naimark-Segal) \cite{TA} construction enables one to formulate (a sector of a) quantum theory within the usual Hilbert space formalism and, given a state $\omega$ on a $C^{*}$- algebra $\mathcal{B}$,
to map the elements of $\mathcal{B}$ to operators acting on a Hilbert space, $\mathcal{H}_{\omega}$.
\begin{theorem}
Let $\mathcal{B}$ be a $C^{*}$-algebra and $\omega$ a continuous positive linear functional on $\mathcal{B}$. Then there exists a unique (up to unitary equivalence) representation of $\mathcal{B}$, $(\pi_{\omega},\mathcal{H}_{\omega})$, on a Hilbert space $\mathcal{H}_{\omega}$ such that
\begin{itemize}
\item $\mathcal{H}_{\omega} $contains a cyclic vector $\xi_{\omega}$, i.e., $\left< \pi_{\omega} (\mathcal{B}) \xi_{\omega}\right>=\mathcal{H}_{\omega} $.
\item For all $a \in \mathcal{B}$, $\omega(a)=(\xi_{\omega},\pi_{\omega}(a)\xi_{\omega})$, where $(\cdot,\cdot)$ is the scalar product on the Hilbert space $\mathcal{H}_{\omega}$.
\end{itemize}
\end{theorem}

Of course, this theorem also holds for abelian $C^{*}$-algebras. 
Every unit vector $\phi\in\mathcal{H}_{\omega}$ defines a state given by $(\phi, \pi_{\omega} (.) \phi)$. As the reader may remember from his/her quantum mechanics course, the  GNS representation associated with $(\mathcal{B}, \omega)$ is irreducible if and only if $\omega$ is a pure state. In this case, every unit vector $\phi \in \mathcal{H}_{\omega}$ defines a pure state, too, because every non-zero vector in the Hilbert space of an irreducible representation is cyclic for $\pi_{\omega}(\mathcal{B})$ and the commutant of this algebra consists of multiples of the identity. If the algebra $\mathcal{B}$ is abelian the GNS representation associated with a pure state is one-dimensional, and hence there are no vectors not colinear with $\xi_{\omega}$ in the GNS space $\mathcal{H}_{\omega}$.  
In contrast, if the theory is quantum, with dim$(\mathcal{H}_{\omega}) \geq 2$, one can always find two noncolinear unit vectors, $\psi$ and $\phi$, in $\mathcal{H}_{\omega}$. If $\omega$ is a pure state then an arbitrary linear combination, $c_{1}{\phi}+ c_{2}{\psi}$, normalized to have norm $= {1}$,  is again a pure state. It follows that, for quantum theories, the space of pure states in a fixed GNS representation has an underlying linear structure that gives rise to a $\textit{superposition principle}$. As we will see shortly, this superposition principle implies that the probabilities of sequences of events associated with pure states do not, in general, obey a 0--1 law, anymore, which is an unmistakable signature of the $non{-}deterministic$ nature of quantum theories (as opposed to realistic theories). Moreover, the fact that, in a quantum theory, $\mathcal{B}$ is non-abelian leads to $\textit{quantum interferences}$, which the reader is familiar with from various well known examples, such as the double slit experiment. Quantum interferences are at the root of the problem that quantum theories can, in general, not be embedded into classical hidden-variables theories and that the notion of $``(mutually\text{ } exclusive)\text{ } events"$ becomes fuzzy; as discussed in Section 4.

\subsubsection{ A brief digression on von Neumann algebras} 
For each $C^{*}$-algebra $\mathcal{B}$ (abelian or not), there is a Hilbert space $\mathcal{H}$ such that $\mathcal{B}$  is isometrically isomorphic to a uniformly closed self-adjoint subalgebra of $B(\mathcal{H})$ (the algebra of all bounded linear operators on $\mathcal{H}$). Von Neumann algebras are a particular type of  $C^{*}$-algebras whose definition is based on this isomorphism. 
\begin{definition}
A von Neumann algebra, $\mathcal{M}$, is a $^*$-subalgebra of the algebra, $B(\mathcal{H})$, of all bounded operators on a Hilbert space $\mathcal{H}$ that is $\sigma$-weakly closed and non-degenerate, i.e.,  $\langle \mathcal{M} \mathcal{H}\rangle=\mathcal{H}$.
\end{definition}

The $\sigma$-weak topology is the topology on $B(\mathcal{H})$ that comes from the isometry  $B(\mathcal{H}) \simeq B_{0}(\mathcal{H})^{**}$, where $B_{0}(\mathcal{H})$ is the set of compact linear operators on $\mathcal{H}$. Von Neumann's {double-commutant} theorem provides another, equivalent definition of von Neumann algebras.
\begin{definition}
A von Neumann algebra, $\mathcal{M}$, is a $^*$-subalgebra of  $B(\mathcal{H})$, with $\mathcal{M}'':=(\mathcal{M}')'=\mathcal{M}$.
\end{definition}
Here $\mathcal{M}'$ is the commutant of $\mathcal{M}$, i.e., the set of elements of $B(\mathcal{H})$ commuting with $\mathcal{M}$. 

Von Neumann algebras are interesting for many reasons. First, for  every representation $(\pi,\mathcal{H})$ of a $C^{*}$-algebra $\mathcal{B}$, $\pi(\mathcal{B})''$ is a von Neumann algebra that is the weak closure of $\pi(\mathcal{B})$. Moreover, using the universal representation of  $\mathcal{B}$, one can show that the second conjugate space of $\mathcal{B}$, $\mathcal{B}^{**}$, is isometric to the von Neumann algebra induced by this universal representation. 
Second, von Neumann algebras with trivial center (factors) are completely classified; (type $I_{n}, n = 1,2,...,\infty$; type $II_{1}$, type $II_{\infty}$; type $III_{\lambda}$, $0 < \lambda \leq 1$).
Finally, and this is the reason why, at this point, we mention von Neumann algebras, quantum systems $S$ with finitely many degrees of freedom can be described in terms of algebras $\mathcal{B}_{S}$ that are type-$I$ von Neumann algebra; e.g., in terms of the  group algebra of some compact group (SU(2), for quantum-mechanical spins) and/or of the Weyl algebra generated by the position- and momentum operators of finitely many particles. A von Neumann algebra is said to be of type $I$ if every non-zero central projection majorizes a non-zero abelian projection in $M$. It is easy to see that $B(\mathcal{H})$ is a type- $I$ von Neumann algebra, because any projection of rank $1$ is abelian. Actually, for every factor, $\mathcal{F}$, of type $I$,  there exists a Hilbert space $\mathcal{H}$ with $\mathcal{F} \simeq{B(\mathcal{H})}$. The theory of direct integrals of von Neumann algebras shows that the direct integral 
$$ \mathcal{M} =\int_{\Xi}^{\oplus} B(\mathcal{H}_{\xi}) $$
 of factors of type $I$ on a standard Borel space $(\Xi,\mu)$ is a type-$I$ von Neumann algebra.
 The  C*-group algebra, $C^{*}(G)$, of a compact group $G$ is isomorphic to the direct sum of all unitary irreducible representations of $G$. As there are at most countably many such representations, all of them finite-dimensional, 
$$C^{*}(G) \simeq \bigoplus_{n \in \mathcal{N}} M_{n}(\mathbb{C}),$$
where $\mathcal{N}$ is a subset of $\mathbb{N}$ (the natural numbers).
A typical example is the quantum-mechanical rotation group, $SU(2)$. Its  C*-group algebra can be used as the algebra of possible events, $\mathcal{B}_{S}$, of a system of quantum-mechanical spins. It is well known that $C^{*}(SU(2))$ is isomorphic to the direct sum of unitary irreducible representations of $SU(2)$, i.e., $C^{*}(SU(2)) \simeq \bigoplus_{s \in \mathbb{N}} M_{2s+1} (\mathbb{C})$. 

Standard results of the theory of direct integrals imply the following facts. 
\begin{itemize}
\item The direct-integral representation of a von Neumann algebra, $\mathcal{M}$, of type $I$ is unique in a rather obvious sense. 
\item Every element of the predual of $\mathcal{M}$ is given by a unique integrable field of elements of the predual of the algebras $B(\mathcal{H}_{\xi})$, $\xi \in \Xi$, i.e., $\omega \in\mathcal{M}^{*}$ is given by
$$\omega=\int_{\Xi}^{\oplus} \omega_{\xi} d\mu(\xi)$$
\end{itemize}

It follows that every normal state, $\omega$, is given by a direct integral of normal states, $\omega_{\xi}=tr(\rho_{\xi}{(.)})$, where the $\rho_{\xi}$ are density matrices on the fibre spaces $\mathcal{H}_{\xi}$. Pure states on $\mathcal{M}$ are given by unit vectors in one of the spaces $\mathcal{H}_{\xi}$, with $\mu$ given by the $\delta-$ function concentrated in $\xi$.\\ 
The dynamics of a system $S$, determined by  $*$-automorphisms, ${\alpha_{t,s}}$, $t,s \in \mathbb{R}$, of $\mathcal{M}$ := $\mathcal{B}_{S}$, is then given in terms of a crossed product of  a measurable field of unitary propagators, $U_{t,s}$, on the fibre Hilbert spaces and a measurable field of Borel isomorphisms, $\Phi_{t,s}$, of strata of $\Xi$.\\
Alternatively, one may define time evolution in terms of maps from the space, $\mathcal{S}_S$, of states of the system $S$ to itself (``Schr\"odinger picture''), rather than in terms of $^{*}$-automorphisms of the algebra $\mathcal{B}_S$ (``Heisenberg picture''). One then assumes that time evolution is given in terms of ``completely positive'' maps from $\mathcal{S}_S$ to itself; but pure states are not necessarily mapped to pure states. This leads one to the theory of quantum stochastic processes and of Lindblad generators, see \cite{Lindblad}, which is often useful for a phenomenological (rather than a fundamental) description of the dynamics of open quantum systems with many degrees of freedom.\\
For the mathematically minded reader we remark that the material in this subsection should really be formulated within the realm of type-$I$ $C^{*}$-algebras, as developed by Glimm \cite{Glimm}, which is a natural framework for the description of quantum systems with finitely many degrees of freedom.

\subsubsection{ Composition of systems and entanglement} 
The composition, $S_1\vee S_2$, of two systems, $S_1$ and $S_2$, is described in terms of the tensor product of the algebras $\mathcal{B}_{S_1}$ and $\mathcal{B}_{S_2}$, i.e., $\mathcal{B}_{S_1\vee S_2} := \mathcal{B}_{S_1} \otimes  \mathcal{B}_{S_2}$. The simplest examples of quantum theories concern systems with $\mathcal{B}_{S_i} \simeq B(H_i)$, where the state spaces $H_i$ are finite-dimensional Hilbert spaces. In contrast to systems described by realistic theories, quantum systems may be $``entangled"$. For instance, $\frac{1}{\sqrt{2}} \left( \phi_{1} \otimes \phi_{2} +\psi_{1} \otimes \psi_{2} \right)$, with  $\phi_1,\psi_1$ a pair of orthogonal unit vectors in $H_1$, $\phi_2,\psi_2$ orthogonal unit vectors in $H_2$, is a pure state of $B(H_1) \otimes B(H_2)$. But its restriction to the algebra $B(H_i)$ of a single subsystem is $not$ a pure state of $S_i$. One says that $S_1$ and $S_2$ are entangled in this state.\\  
Many interesting mathematical problems arise in the study of composed quantum systems. As an example, we mention the problem of \textit{``quantum marginals''}. We continue to consider systems for which $\mathcal{B}_{S_i}\simeq B(H_i)$, with $H_i$ finite-dimensional, for $i=1,...n$. Then
$$\mathcal{B}_{S_1\vee...\vee S_n} := \bigotimes_{i=1}^{n} B(H_i) \simeq B \left( \otimes_{i=1}^{n} H_i \right)$$
Let $\rho$ be a density matrix on  $\otimes_{i=1}^{n} H_i $, and let $\rho_i$ be its $\text{i}^{\text{th}}$ marginal, defined by 
$$\text{tr}(\rho_i a):= \text{tr} (\rho (1 \otimes\cdot\cdot\cdot\otimes 1\otimes \underbrace{a}_{i}  \otimes 1\otimes \cdot\cdot\cdot \otimes 1)), $$
for all $a \in B(H_{i}) $.
It is natural to ask the following question: Given $\rho$ and $\rho_1,...,\rho_n$, under what conditions on (the spectra of) the density matrices $\rho$ and $\rho_i$, $i = 1,...,n$, is $\rho_i$ the $\text{i}^{\text{th}}$ marginal of $\rho$? This difficult question has been answered by Klyachko in 2004, see $\cite{Klyachko}$. His main result is described in Appendix \ref{A1}. (The example where $\rho$ is a pure state and $n=2$ is, of course, well known and elementary.)

\subsubsection{The Kochen-Specker theorem on the absence of hidden variables}
In this subsection, we present a short overview of  the Kochen-Specker theorem, a ``no-go theorem'' for hidden-variables theories. 

In 1935, Einstein, Podolsky and Rosen \cite{EPR} proposed to extend quantum theories to realistic theories. The question immediately arises as to wether this is possible. Here we recall a result due to Kochen and Specker concerning the impossibility of  ``hidden-variables theories''.

We consider a system $S$ described by a quantum theory with a non-commutative $*$-algebra of observables $\mathcal{A}^{Q}_{S}\subseteq\mathcal{B}^{Q}_{S}$, (the $C^{*}$-algebra of possible events in S). We suppose that $\mathcal{B}^{Q}_{S}$ is unital. 

Due to Gel'fand's isomorphism, a realistic theory describing $S$ would have an algebra of possible events of the form $\mathcal{B}_{S}^{C} \simeq \mathcal{C}_{0}(\Omega)$, where $\Omega$ is a compact Hausdorff space, and $\mathcal{C}_{0}(\Omega)$ is the set of continuous functions on $\Omega$. States on  $\mathcal{B}_{S}^{C}$ are probability measures. For every probability measure, $\mu$, on $\Omega$, the GNS representation, $\pi$, associated with $(\mathcal{B}_{S}^{C}, \mu)$ is realized by multiplication operators on the Hilbert space of square-integrable functions, $L^{2}(\Omega, \mu)$, and the corresponding cyclic vector is the constant function $\equiv1$ on $\Omega$. The von Neumann algebra generated by $\mathcal{B}_{S}^{C}$ in this representation is given by $L^{\infty} (\Omega, \mu)\supseteq\mathcal{C}_{0}(\Omega)$. Hidden-variables theories may thus be viewed as realistic theories specified by an algebra of bounded measurable functions, $\mathcal{F}_{\Omega}$, on a measure space $(\Omega,\sigma)$, where $\sigma$ is a $\sigma$-algebra. 

With this in mind, we may attempt to construct an embedding of a quantum theory in a realistic theory in the following way:  Let  $A=A^{*} \in \mathcal{A}^{Q}_{S} \subseteq B(\mathcal{H}_{S})$ be an observable (i.e., an operator corresponding to some physical quantity of $S$), and let $P_{A}$ denote the spectral projections of $A$. For any $\psi \in \mathcal{H}_{S}$, $(\psi, P_{A}(\cdot) \psi)$ is a probability measure on the spectrum, $\sigma(A)$, of $A$. We suppose that  a random variable $\alpha_{A}: \Omega \rightarrow \mathbb{R}$, $\alpha_{A} \in \mathcal{F}^{\mathbb{R}}_{\Omega}$, can be associated with $A$, where $ (\Omega,\sigma)$ is a measure space independent of the observable $A$. We also suppose that we can  associate a probability measure $\mu_{\psi}$ on $\Omega$ to any  vector $\psi \in \mathcal{H}_{S}$, with the property that
$$(\psi, P_{A}(\Delta) \psi)=\mu_{\psi}(\alpha_{A}^{-1}(\Delta))$$
for an arbitrary measurable set $\Delta   \subset \sigma(A)$. This would imply that the quantum theory and the realistic theory it is embedded in predict the same probability distributions for the measured values of the observable A. Because any real function, $f$, of an observable is again an observable, it is natural to require that $f(\alpha_{A})=\alpha_{f(A)}$, for every such $f$. 
We are led to the following definition of a hidden-variables embedding, $ \mathcal{F}_{\Omega}$, of $\mathcal{A}^{Q}_{S}$:
\begin{definition}
Hidden-variables embedding 
\end{definition}
Let $\Omega$ be a measure space. A hidden-variables embedding, $(\Omega, \mathcal{F}_{\Omega})$, of $\mathcal{A}^{Q}_{S} \subset B(\mathcal{H}_{S})$ is defined by specifying maps
\begin{eqnarray*}
A=A^{*} \in \mathcal{A}^{Q}_{S} &\mapsto & \alpha_{A}  \in \mathcal{F}^{\mathbb{R}}_{\Omega},\\
\psi \in \mathcal{H}_{S} &\mapsto &  \mu_{\psi} \in \text{prob}_{\Omega},
\end{eqnarray*}
where $\text{prob}_{\Omega}$ is the collection of probability measures on $\Omega$, with the properties
\begin{eqnarray*}
&(1)& (\psi, P_{A}(\Delta) \psi)=\mu_{\psi}(\alpha_{A}^{-1}(\Delta)), \text{ for all } \Delta   \subset \sigma(A);\\
&(2)& f(\alpha_{A})=\alpha_{f(A)}, \text{ for any real continuous function}  f.
\end{eqnarray*}

Kochen and Specker have proven the  following theorem.
\begin{theorem}(Kochen-Specker).
If $\mathcal{A}_{S}^{Q}$ = $B(\mathcal{H}_{S})$, with $\dim(\mathcal{H}_{S}) \geq 3$, it is impossible to find a hidden-variables embedding of $\mathcal{A}_{S}^{Q}$ into $(\Omega,\mathcal{F}_{\Omega})$. 
\end{theorem}
Numerous proofs of this theorem can be found in the litterature. The reader may enjoy consulting the original paper $\cite{Kochen-Specker}$. For a recent, simple proof see $\cite{Straumann}$, and references given there. 

\subsubsection{Correlation matrices and Bell's inequalities} 
Bell's inequality \cite{Bell} has played a very prominent role in much recent theoretical and experimental work concerning the foundations of quantum science. It therefore should appear on stage in notes like these. In Appendix \ref{A2}, we briefly review Tsirelson's work on Bell's inequalities. 
Here, we merely recall a variant of Bell's inequality due to Clauser, Horne, Shimony and Holt \cite{Clauser}, which is a special case of the general framework outlined in Appendix \ref{A2}.\\
We consider two observers, A (for ``Alice'') and B (for ``Bob''), who are measuring spins or helicities in a system of two particles (e.g., electrons or, more realistically, photons) of spin $\frac{1}{2}$ or helicity $\pm 1$, respectively; (in the following discussion, we speak of spin).
A particle source emits one of these particles in the direction of A and the other one in the direction of B. Alice measures the component of the spin of one particle along the directions $\mathbf{t}$ or $\mathbf{u}$, while Bob measures the component of the spin of the other particle along  $\mathbf{v}$ or $\mathbf{w}$.\\
Let us first imagine that there is a $realistic$ theory describing these spin measurements. Denoting the component of the spin of a particle along $\textbf{u}$ by $\sigma_{\textbf{u}}$, one then observes that
$ \sigma_{\textbf{t}}$, $\sigma_{\textbf{u}}$, $\sigma_{\textbf{v}}$ and $\sigma_{\textbf{w}}$ are random variables on some measure space $(\Omega,\mu)$ taking the values $\pm 1$. 
(For simplicity, these quantities are rescaled by a factor $2/{\hbar}$, so that their values are $\pm 1$, rather than $\pm {\hbar}/2$). It is immediate to see that, for an arbitrary $\omega$ in $\Omega$,
\begin{equation}
\label{classic}
\sigma_{\textbf{t}} (\omega) \sigma_{\textbf{v}} (\omega) + \sigma_{\textbf{t}} (\omega) \sigma_{\textbf{w}} (\omega)+\sigma_{\textbf{u}} (\omega) \sigma_{\textbf{v}} (\omega)-\sigma_{\textbf{u}} (\omega) \sigma_{\textbf{w}} (\omega)=\pm 2
\end{equation}

Integrating over $\Omega$, we find that the correlations 
$$\langle \sigma_{\textbf{t}}  \sigma_{\textbf{v}}  \rangle:=\int_{\Omega }\sigma_{\textbf{t}} (\omega) \sigma_{\textbf{v}} (\omega) d\mu(\omega)$$
satisfy the inequalities
\begin{equation}
\label{classic2}
\vert \langle \sigma_{\textbf{t}}  \sigma_{\textbf{v}}  \rangle + \langle \sigma_{\textbf{t}}  \sigma_{\textbf{w}}  \rangle + \langle \sigma_{\textbf{u}}  \sigma_{\textbf{v}}  \rangle -\langle \sigma_{\textbf{u}}  \sigma_{\textbf{w}}  \rangle \vert \leq 2
\end{equation}
These inequalities characterize a polytope of $classical$ correlation matrices (cf. Appendix \ref{A2}).\\
 It is well known, however, that in actual spin- or, rather, photon-polarization experiments inequalities (\ref{classic2}) are violated (see  \cite{Haroche} and refs.), \textit{as predicted by quantum mechanics}. Indeed, preparing the two particles in an appropriate pure, but entangled state $\vert \psi \rangle \in \mathcal{H}_A \otimes \mathcal{H}_{B}$, it is not difficult to show that, for a certain choice of the angles between the axes $\mathbf{t}$,  $\mathbf{u}$, $\mathbf{v}$ and $\mathbf{w}$,
$$\langle \psi \vert S_{\textbf{t}} \otimes S_{\textbf{v}}  + S_{\textbf{t}}  \otimes S_{\textbf{w}}+ S_{\textbf{u}}  \otimes S_{\textbf{v}}-  S_{\textbf{u}}  \otimes  S_{\textbf{w}}  \vert \psi \rangle=2 \sqrt{2} $$
where $S_{\textbf{u}}:=\left( u_x \sigma_x +  u_y \sigma_y + u_z \sigma_z \right)$ is the rescaled spin operator along the $\textbf{u}$ axis (with $\sigma_x$, $\sigma_y$ and  $\sigma_z$ the usual Pauli matrices). 
This violation of inequality (\ref{classic2}) clearly represents a ``no-go theorem'' for hidden-variables theories.\\
For a more detailed discussion of this topic, the reader is referred to Appendix \ref{A2} and   references given there.

\subsection{Quantum mechanics and indeterminism -- the no-signaling lemma}
In this section, we present some simple arguments explaining why quantum theories are intrinsically non-deterministic. \\
We begin our discussion by considering a quantum-mechanical system given as the composition of a subsystem, $S$, to be studied experimentally and another subsystem, $E$, the measuring apparatus, designed to measure certain physical quantities pertaining to $S$. A measurement of a physical quantity is supposed to trigger an ``event'', which, according to Section 2, we can identify with a spectral projection of the selfadjoint operator representing the physical quantity that is measured. 
Suppose that $S\vee{E}$ is prepared in a pure state $\omega$, and that the measuring apparatus $E$ is triggering a sequence, $P_1,...,P_n$, of events that take place in this order, i.e., $P_1$ is the first event registered, $P_2$ the second one, etc. and $P_n$ the last one. 
As will be explained in Section 4, the state $\omega$ enables us to predict the probability that the sequence $P_1,...,P_n$ of possible events  in $S$ is actually observed in an experiment. The crucial observation is then that if the algebra $\mathcal{A}_{S}$ is non-abelian then there are states $\omega$ predicting probabilities for certain sequences $P_1,...,P_n$ of possible events to be observed that do $not$ obey a ``0--1 law'', \textit{even if the state $\omega$ is pure}. 
This shows that the theory is $not$ deterministic.\\
In order for this argument to be convincing, we would have to explain why the system $S\vee{E}$ can be prepared in pure states $\omega$ that do $not$ predict ``0--1 laws'' for certain sequences of events, and why such states are obtained as outputs of physical processes and contain maximal information on the system. Luckily, the theory of preparation of quantum-mechanical systems in rather arbitrary pure states has seen important advances, in recent times, and hence the argument indicated above can be made into a proof of indeterminism in quantum mechanics. 
However, it may be useful to present a more concrete argument that, in addition, clarifies some further salient features of the quantum mechanics of composed systems prepared in an entangled state.\\
We consider a system $S = S_1\vee{S_2}$ consisting of two subsystems, $S_1$ and $S_2$, that are prepared in an entangled state. We imagine that, after preparation, the two systems are sufficiently far separated from each other that they evolve in time more or less independently. We propose to show that if the outcome of a measurement of a physical quantity pertaining to subsystem $S_1$ could be predicted (with certainty) as the result of the unitary time evolution of the initial state of the system $(S_1\vee{S_2})\vee{E}$, consisting of $S_1\vee{S_2}$ composed with some experimental equipment $E$, the resulting state of the entire system would not reproduce the standard quantum-mechanical correlations between measurement outcomes in subsystems $S_1$ and $S_2$ when some physical quantity pertaining to $S_2$ is measured later on. Our reasoning process is based on a result in \cite{FPS}, called the ``no-signaling lemma'', that we sketch below for a simple example.

We consider a static source (e.g., a heavy atom bombarded with light pulses) that can emit a pair of electrons prepared in a spin-singlet state, with orbital wave functions that evolve into conical regions opening to the left and the right of the source, denoted by $L$ and $R$, respectively, under the two-particle time evolution -- up to exponentially small tails extending beyond these conical regions. Sources with approximately these properties can be manufactured. The experimental setup is indicated in the figure below. 

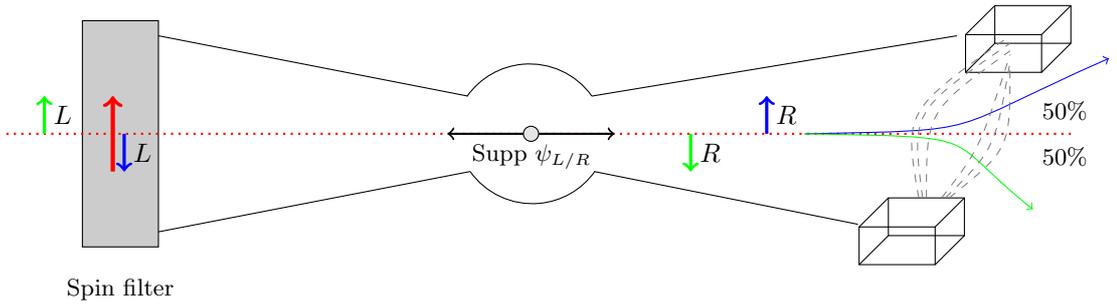
\begin{center}
\begin{tikzpicture}
\draw  [fill=black!20](-1,-1.5) rectangle (0,1.5); 
\draw  [ultra thick,color=red,->](-0.6,-0.5) -- (-0.6,0.5); 
\draw[thick,dotted,color=red,-] (-2,0) -- (12,0);
\draw[thick,->] (5,0) -- (6,0);
\draw[thick,->] (4.8,0) -- (3.8,0);
\draw[fill=black!10] (4.9,0) circle (0.1);
\draw (5.7,0.5) arc (35:145:1);
\draw (4.1,-0.5) arc (-145:-35:1);
\draw[-] (0,1.3) -- (4.07,0.5);
\draw[-] (0,-1.3) -- (4.1,-0.5);
\draw[-] (5.7,0.5) -- (10.5,1.3);
\draw[-] (5.73,-0.5) -- (9.2,-1.2);

\draw[very thick,color=blue,->] (8,0) -- (8,0.5);
\draw[very thick,color=green,->] (7,0) -- (7,-0.5);
\draw[very thick,color=green,->] (-1.5,0) -- (-1.5,0.5);
\draw[very thick,color=blue,->] (-0.45,0) -- (-0.45,-0.5);
\draw(8,0.25) node[right]{$R$};
\draw(7,-0.25) node[right]{$R$};
\draw(-1.5,0.25) node[right]{$L$};
\draw(-0.45,-0.25) node[right]{$L$};
\draw(4,-0.3) node[right]{\small Supp $\psi_{L/R}$};
\draw(11.5,+0.3) node[right]{\small $50\%$};
\draw(11.5,-0.3) node[right]{\small $50\%$};
\draw(-0.5,-1.8) node[below]{\small Spin filter};

\begin{scope}[rotate=0,xshift=9.6cm,yshift=-1.35cm]
\draw (0,0,0)--(1,0,0)--(1,0.5,0)--(0,0.5,0)--cycle; 
\draw  (0,0,1)--(1,0,1)--(1,0.5,1)--(0,0.5,1)--cycle; 
\draw (0,0,0) -- (0,0,1); 
\draw (1,0,0) -- (1,0,1); 
\draw (1,0.5,0) -- (1,0.5,1); 
\draw(0,0.5,0) -- (0,0.5,1); 
\end{scope}

\begin{scope}[rotate=0,xshift=11cm,yshift=1.2cm]
\draw (0,0,0)--(1,0,0)--(1,0.5,0)--(0,0.5,0)--cycle; 
\draw  (0,0,1)--(1,0,1)--(1,0.5,1)--(0,0.5,1)--cycle; 
\draw (0,0,0) -- (0,0,1); 
\draw (1,0,0) -- (1,0,1); 
\draw (1,0.5,0) -- (1,0.5,1); 
\draw(0,0.5,0) -- (0,0.5,1); 
\end{scope}

\draw [->,color=blue](8.5,0) .. controls (11,0.05) and (10.25,0) .. (12.5,1);
\draw [->,color=green](8.5,0) .. controls (11,-0.05) and (10.25,0) .. (11.5,-1);

\draw [-,dashed,color=gray](10.4,-0.85) .. controls (11.2,-0.2) and (11.2,-0.2) .. (11.2,0.85);
\draw [-,dashed,color=gray](10.35,-0.85) .. controls (11,-0.2) and (11,-0.2) .. (11.2,0.95);
\draw [-,dashed,color=gray](10.3,-0.85) .. controls (10.8,-0.2) and (10.8,-0.2) .. (11.2,1.05);
\draw [-,dashed,color=gray](10.1,-0.85) .. controls (10.1,0.3) and (10.1,0.3) .. (11.2,1.1);
\draw [-,dashed,color=gray](10.05,-0.85) .. controls (9.9,0.3) and (9.9,0.3) .. (11.2,1.2);
\draw [-,dashed,color=gray](10.,-0.85) .. controls (9.8,0.3) and (9.8,0.3) .. (11.2,1.3);

\end{tikzpicture}
\end{center}

Let $L^{2}(\mathbb{R}^{3})$ denote the Hilbert space of square-integrable functions on $\mathbb{R}^{3}$ -- orbital wave functions of a single electron --, and let $\mathbb{C}^{2}$ be  the state space of the spin of an electron. In $\mathbb{C}^{2}$, we choose the standard basis, $\vert \uparrow  \rangle$ and $\vert \downarrow \rangle$, of normalized eigenvectors of the 3-component, ($\hbar/2)\sigma_{3}$, of the electron spin operator. We denote the Hilbert space of a system, $S$, consisting of two electrons by $\mathcal{H}_{S}:=\mathcal{A}(L^{2}(\mathbb{R}^{3}) \otimes \mathbb{C}^{2})^{\otimes2}$, where $\mathcal{A}$ is the projection onto anti-symmetric wave functions implementing Pauli's exclusion principle. 
If the two electrons are prepared in a spin-singlet state ($S_{tot} = 0$) the total spin wave function is anti-symmetric, while the total orbital wave function of the electrons is symmetric under exchange of the electron variables. 
We choose a one-electron orbital wave function, $\vert\psi_L\rangle$, evolving into $L$ under the free time evolution (electron-electron interactions are neglected for simplicity) and one, $\vert\psi_R\rangle$, evolving into $R$ -- except for small tails, as mentioned above -- with $\langle \psi_R \vert\psi_L\rangle\approx{0}$. The two-electron system $S$ is then assumed to be prepared in an initial state given by the unit vector
\begin{equation}
\Psi:=\frac{1}{2} \left( \vert\psi_L\rangle \otimes \vert\psi_R\rangle + \vert\psi_R\rangle \otimes \vert\psi_L\rangle\right) \bigotimes \left( \vert \downarrow \rangle \otimes \vert \uparrow \rangle - \vert \uparrow \rangle \otimes \vert \downarrow \rangle\right)
\end{equation}

The experiment sketched in the above figure is designed to measure the 3-component of the spin of the electron evolving into $L$ with the help of a spin filter. This filter absorbs an electron penetrating it if its spin is ``down'' (i.e., if its spin wave function is given by $\vert \downarrow \rangle$), and it lets the electron pass through it if its spin is ``up'' (i.e., if its spin wave function is given by $\vert \uparrow \rangle$). Ferromagnetic metallic films magnetized in the 3-direction can be used as such filters. 
Far away from the source, in the region $R$, a Stern-Gerlach-type experiment may be performed to measure some component of the spin of the electron evolving into $R$, after the spin measurement on the electron evolving into $L$ has been completed. The entire experimental equipment used to do these two measurements represents a quantum-mechanical system denoted by $E$. The total system to be analyzed here is the composition, $S\vee{E}$, of $S$ and $E$. For simplicity, let us suppose that the 3-component of the spin is measured in the region $R$. We denote by $\uparrow_{L/R}$ the event that an electron with spin ``up'' is observed on the left/right. A similar notation is used for spin ``down''.

\textbf{Fact:} We assume that the system $S$ is prepared in a spin-singlet state of the form described above. Then if an experimentalist, called Alice, observes $\uparrow_{L}$ in her laboratory she predicts that her colleague, called Bob, will observe $\downarrow_{R}$ -- and, for all we know about such experiments, he sure will. Similarly, if  $\downarrow_{L}$ is observed by Alice she predicts that Bob will observe $\uparrow_{R}$.

Let us assume, temporarily, that the quantum-mechanical description of $S\vee{E}$ has a realistic interpretation. If this assumption were legitimate then the experimental fact described above would have to emerge as the consequence of some unitary time evolution applied to ``typical'' initial states of $S\vee{E}$. It is shown in \cite{FPS} that, under physically very plausible assumptions on the interactions between the electrons and the spin filter, this is impossible. (The Stern-Gerlach experiment is described, for simplicity, by an external magnetic field turned on in the region $R$, very far away from the source.)

Let $\mathcal{H}_{S}$ denote the Hilbert space of state vectors of the two-electron system $S$ and $\mathcal{H}_{E}$ the one of the spin filter $E$. We choose an initial state, $\Phi\in\mathcal{H}_{S}\otimes\mathcal{H}_{E}$, of the composed system at time $t=0$ of the form
\begin{equation}
\Phi:=\sum_{\alpha}\Psi_{\alpha}\otimes\chi_\text{filter}^{\alpha},
\end{equation}
where the vectors $\Psi_{\alpha}$ are spin-singlet two-electron wave functions of the form of the vector $\Psi$ defined in Eq. (1), and the vectors $\chi_\text{filter}^{\alpha}$ all belong to the same sector $\mathcal{H}_{E}$ of the spin filter.

The dynamics of the composed system $S\vee{E}$ is given by a Hamiltonian
\begin{equation}
H:=H_0+H_I,
\end{equation}
acting on the space $\mathcal{H}_{S}\otimes\mathcal{H}_{E}$, where $H_0$ is the Hamiltonian of the system before the electrons are coupled to the spin filter, and $H_I$ describes the interactions between the electrons and the filter. The operator $H_I$ is localized in a compact region around the filter (in a sense made precise in \cite{FPS}). The time evolution of $\Phi$ in the Schr\"odinger picture is given by
$$\Phi_t:=e^{-itH} \Phi.$$
By $\mathcal{\bf{S}}_{R}$ we denote the spin operator localized in the region $R$; see \cite{FPS}.

\begin{lemma} (``No-signaling lemma'')

Under certain physically plausible hypotheses on the operator $H_I$ and on the choice of the initial state $\Phi$ of $S\vee{E}$ described in \cite{FPS},
\begin{equation}
\label{right}
(\Phi_t, \mathcal{\bf{S}}_{R} \Phi_t)\approx0,
\end{equation}
 for all $t> 0$; (the estimate on the left side being uniform in $t$).
\end{lemma}

Assuming that, for a judicious choice of the state vectors $\lbrace\chi_\text{filter}^{\alpha}\rbrace$ of the spin filter, the states $\lbrace\Phi_{t}\rbrace_{t>0}$ describe an electron evolving into the region $L$ that passes through the spin filter, i.e., that $\{\Phi_{t}\}_{t>0}$ predicts the event $\uparrow_{L}$ to happen, we run into a contradiction between the statement of Lemma 10 and the experimental $\bf{Fact}$ concerning quantum-mechanical correlations, namely that the event $\uparrow_{L}$ is overwhelmingly correlated with the event $\downarrow_{R}$, which would imply that
\begin{equation}
\label{wrong}
(\Phi_t, \mathcal{\bf{S}}_{R} \Phi_t)\approx-\frac{\hbar}{2} \bf{e}_{3},
\end{equation}
in contradiction with Eq. (\ref{right}). (Similar reasoning applies when $\uparrow_{L}$ is replaced by $\downarrow_{L}$.)\\
Put differently, assuming that the propagator $e^{-itH}$ of $S\vee{E}$ commutes with the 3-component of the total spin operator of the two electrons (as one would guess from symmetry considerations), the assumption that the states $\lbrace\Phi_{t}\rbrace_{t>0}$ predict the event $\uparrow_{L}$ (or $\downarrow_{L}$) to happen, together with the conservation of the 3-component of the total spin, contradicts Eq. (6). 
Thus, a realistic interpretation of quantum mechanics, in the sense that the time evolution of pure states in the Schr\"odinger picture predicts which events will happen, is apparently untenable. Quantum mechanics only predicts $probabilities$ of events, even if the initial state of the entire system is pure. This is $not$ an expression of incomplete knowledge of the system, but is an intrinsic feature of the theory. It will be explained in Section 4 how probabilities of (sequences of) events in a general quantum-mechanical system are calculated, given a state of the system. In the particular example just studied, all that quantum mechanics predicts is that if the two electrons are prepared in a spin-singlet (i.e., entangled, but pure) state then the events $\uparrow_{L}$ and $\downarrow_{L}$ both have probability $\frac{1}{2}$. 

Note that, in Lemma 5, no assumption of Einstein causality, whose use is totally out of place in non-relativistic quantum mechanics, or anything like that has to be invoked.

Quite apart from its consequences concerning indeterminism, we think that Lemma 5 is of some indepedent interest, and this is why we are reporting it here. (The techniques used to prove Lemma 5 can be used, for example, to establish upper bounds on the amount of dynamically generated entanglement between a physical system and a piece of equipment located far away from the system.)

\subsection{Quantization and classical (mean-field) limit}
In this section, we recall the Heisenberg-Dirac recipe of how to ``quantize'' a classical Hamiltonian system with an affine phase space and the reverse process of passing to a ``classical'' Hamiltonian regime of quantum theory. The classical limit of wave mechanics was first analyzed by Schr\"odinger, in 1926, using coherent states; see \cite{Schrodinger}. His analysis was put on rigorous mathematical grounds by Hepp \cite{Hepp} and followers. The up-shot of their results is that ``time evolution and quantization commute, up to error terms that tend to 0, as the deformation parameter (conventionally $\hbar$) approaches 0''. 
For more recent results in this direction and references to the literature, see \cite{FRKP}, \cite{FRKP'}. A detailed exposition of these matters goes beyond the scope of these notes.

 The main purpose of this section is to explain how atomistic theories of matter can be interpreted as the  ``quantization'' of continuum theories. We do this by considering some rather simple, but physically important examples. The reason for sketching these things is to convince the reader that continuum theories of matter, such as the Gross-Pitaevsky theory of a Bose gas, tend to be realistic and deterministic, even if $\hbar$ appears in the equations, and that the loss of realism and determinism in non-relativistic quantum mechanics arises as the result of a conspiracy between the quantum-mechanical nature of matter and its atomistic constitution. Our analysis is based on results in \cite{Hepp}, \cite{FRKP}, \cite{FRKP'}.

It is a well known fact, first established in the 1970's in \cite{Braun}, \cite{Neunzert}, that the mean-field limit of the Newtonian mechanics of many weakly interacting point-particles is given by the Vlasov theory of interacting gases, which turns out to be a $Hamiltonian$ continuum theory of matter.  In \cite{FRKP}, \cite{FRKP'}, the converse has been established:  The classical Newtonian mechanics of systems of many interacting point-particles, i.e., an atomistic theory of matter, can be viewed as the ``quantization'' of Vlasov theory.\\
To explain these things, we begin by considering the Newtonian equations of motion of a system of $n=1,2,...$ identical interacting point-particles,
\begin{equation}
\label{mot}
\frac{d^2q_i}{dt^2} = - g \sum_{j \neq i } \nabla W (q_i - q_j) - \nabla V (q_i),
\end{equation}
where $q_{i}\equiv q_{i}(t)\in\mathbb{R}^{3}$ is the position of the $i^{th}$ particle, $i=1,...,n$, $W$ is a (e.g., bounded and smooth) translation-invariant two-body potential, $g$ is a coupling constant, and $V$ (a smooth and polynomially bounded function on $\mathbb{R}^{3}$) is the potential of an external force; the mass of the particles is set to 1.
To a solution $(q_1 (t),...,q_n(t))$ of the classical equations of motion (\ref{mot}), there corresponds an empirical measure on the one-particle phase space $\mathbb{R}^{6}$ given by
\begin{equation}
\label{emp}
\mu_{n} (t)=\frac{1}{N} \sum_{i=1}^{n} \delta_{(q_{i}(t),v_i(t))},
\end{equation}
where $v_{i}(t):=\frac{dq_{i}(t)}{dt}$ is the velocity of the $i^{th}$ particle.
The number of particles, $n$, is related to the coupling constant $g$ by $n\propto{g^{-1}}$. We set $n=\nu N$, where $N:=N_A$ is Avogadro's number and $\nu$ is the number of moles in the gas. For simplicity, we identify $g$ with $1/N$.

\begin{theorem}
(Mean-Field Limit, \cite{Braun}, \cite{Neunzert})
We assume that $V,W \in \mathcal{C}^{2}(\mathbb{R}^{3})$. Let  $(q_1 (t),...,q_n(t))$  be a solution of (\ref{mot}). Then the measure $\mu_{n} (t)$ defined in (\ref{emp}) is a weak solution of the Vlasov-Poisson equation
\begin{equation}
\label{Vlas}
\frac{\partial \mu}{\partial t} + v \cdot \nabla_{q} \mu - (\nabla_q V_{eff} \left[ \mu \right]) \cdot \nabla_{v} \mu =0
\end{equation}
where $(q,v)\in \mathbb{R}^{6}$, with $v$ the velocity of an element of gas at the point $q$, and
$$V_{eff} \left[ \mu \right] (q) := V(q)+ \int_{\mathbb{R}^{6}} W(q-r) d\mu (r,v)$$
Moreover, if $\mu_{n} (t=0) \underset{n \rightarrow \infty}{\rightarrow} \mu_{0}$, in the weak sense, then  $\mu_{n} (t)\underset{n \rightarrow \infty}{\rightarrow}  \mu (t)$,  in the weak sense, for any $t \geq 0$, where $\mu(t)$ is a weak solution of the Vlasov-Poisson equation with initial condition $\mu_{0}$.
\end{theorem} 
Next, we  explain how to ``quantize'' Vlasov theory to arrive at the Newtonian mechanics of $n=1,2,...$ identical interacting point-particles. We assume that the measure $\mu$ describing a state of Vlasov theory is absolutely continuous with respect to Lebesgue measure $d^{3}qd^{3}v$, so that it has a non-negative density (Radon-Nikodym derivative) $f\in L^{1}(\mathbb{R}^{6})$, i.e.,  
$$d\mu:=d\mu^{f}=f(q,v)d^3 q d^3 v,  \text{ with} \int_{\mathbb{R}^{6}} d \mu^{f}=\nu.$$
Then the solution $\mu^{f}(t)$ of (\ref{Vlas}), with initial condition $\mu^{f}$, has a non-negative density $f_t\in L^{1}(\mathbb{R}^6)$. Assuming that $f \in \mathcal{C}^{1}(\mathbb{R}^6)$, the measure $\mu^{f_t}$ is a strong solution of the Vlasov-Poisson equation with initial condition $\mu^{f}$, and $f_t  \in \mathcal{C}^{1}(\mathbb{R}^6)$.

Because a density $f$ is non-negative, it can be written in the form $f(q,v)=\bar{\alpha}(q,v) \alpha(q,v) $, for some (complex-valued) half-density $\alpha\equiv\alpha^{f} \in L^{2}(\mathbb{R}^6)$, with $\bar{\alpha}$ the complex conjugate of $\alpha$. Clearly, the (local) phase of the half-density $\alpha$ is arbitrary and does not have any physical meaning. Formulating Vlasov theory in terms of half-densities thus leads to a gauge symmetry of the second kind, 
\begin{equation}
\label{gauge}
\alpha^{\sharp}(q,v)\mapsto {\text{exp}}[\pm i\phi(q,v)]\alpha^{\sharp}(q,v),
\end{equation}
$\phi(q,v)\in\mathbb{R}$, $\alpha^{\sharp}= \alpha$ or $\bar{\alpha}$,
that leaves all physical quantities invariant.
From now on, we assume that $\alpha$ belongs to the complex Sobolev space $\Gamma_{V}:=H^{1}(\mathbb{R}^{6})$, which we interpret as an infinite-dimensional complex affine phase space. The pairs $(\alpha, \bar\alpha)$ can be interpreted as complex coordinates on $\Gamma_{V}$. 
Phase space $\Gamma_{V}$ is equipped with the symplectic 2-form
$$\sigma:= i \int d^3q d^3v \text{ }  d\bar{\alpha} (q,v) \wedge d\alpha (q,v), $$
which gives rise to the Poisson brackets 
\begin{equation}
\lbrace \alpha^{\sharp} (q,v), \alpha^{\sharp} (q',v') \rbrace= 0, \text{ } \text{ } \lbrace \alpha(q,v), \bar{\alpha}(q',v') \rbrace= i \delta(q-q') \otimes \delta(p-p').
\end{equation}

We introduce a Hamilton functional
\begin{equation}
\begin{split}
\label{Hamilton}
H(\alpha,\bar{\alpha}):=
i \int d^{3}q d^{3}v \text{ } &\bar{\alpha} (q,v) \left[ \frac{}{}(-v \cdot \nabla_q + \nabla_q V \cdot \nabla_v  + \right.\\
& \left. \left( \int d^{3}q' d^{3}v' \text{ }\nabla W (q-q') \mid \alpha(q',v')\mid^{2} \right) \cdot \nabla_v\right] \alpha(q,v)
\end{split}
\end{equation}
The Hamiltonian equations of motion are
\begin{equation}
\label{motion}
\frac{d}{dt}\alpha^{\sharp} _{t} (q,v):=\lbrace H, \alpha^{\sharp}_{t} (q,v)\rbrace
\end{equation}
{\bf{Fact:}} An easy but important observation is that if $(\alpha_t,\bar{\alpha}_t)$ is a solution of the Hamiltonian equations of motion (\ref{motion}) then 
$$f_t(q,v):=\bar{\alpha}_{t}(q,v)\cdot{\alpha}_{t}(q,v)$$
is a strong solution of the Vlasov-Poisson equation, in the sense that the measure $\mu^{f_t}$ solves Eq.(\ref{Vlas}). Invariance of the Hamilton functional $H$ under global gauge transformations implies that the quantity
$$\int\bar{\alpha}_{t}(q,v)\cdot \alpha_{t}(q,v) \text{ }d^{3}qd^{3}v\equiv \int f_{t}(q,v)\text{ } d^{3}qd^{3}v$$ is conserved, as follows from Noether's theorem. The Hamiltonian vector field associated with the functional on the left side generates global phase transformations.
Thus, by factorizing densities, $f$, into a product of half-densities, $\alpha$, with their complex conjugates, $\bar{\alpha}$, we have succeeded in finding a Hamiltonian formulation of Vlasov theory on an infinite-dimensional complex phase space. This formulation gives rise to the local gauge symmetry described in Eq. (\ref{gauge}).
Physical quantities, i.e., ``observables'', must be $invariant$ under the gauge transformations (\ref{gauge}). They only depend on the densities $f$, but $not$ on the phases of the half-densities $\alpha$. Hence they have the form
$$A(f;\underline{w})=\sum_{n=1}^{\infty} \int  w_n(\underline{q}_n,\underline{v}_n) \prod_{i=1}^{n}  f(q_i,v_i)d^{3}q_{i} d^{3}v_{i}$$
where $(\underline{q}_n,\underline{v}_n):=(q_1,...,q_n,v_1,...,v_n)$, and the kernels $w_n$ are continuous functions on $\mathbb{R}^{6n}$ whose sup-norms tend rapidly to 0, as $n\rightarrow\infty$. In complex coordinates $(\alpha,\bar{\alpha})$,
\begin{equation}
\label{obs}
A(\alpha, \bar{\alpha};\underline{w})=\sum_{n=1}^{\infty} \int w_n(\underline{q}_n,\underline{v}_n) \prod_{i=1}^{n}  \vert \alpha(q_i,v_i) \vert ^{2} d^{3}q_{i} d^{3}v_{i}.
\end{equation}
These observables generate an abelian algebra, $\mathcal{A}^{C}$, of functions on the phase space $\Gamma_{V}$.

We proceed to $quantize$ Vlasov theory, using its Hamiltonian formulation. Since the phase space $\Gamma_{V}$ of this theory is an affine complex space, we can follow Dirac's prescription and replace \textit{Poisson brackets} by $iN \times {commutators}$, with $\frac{1}{N}$ playing the role of $\hbar$. The half-density $\alpha$ is then replaced by an annihilation operator, $a_{N}$, and $\bar{\alpha}$ by a creation operator, $a_{N}^{*}$,  and we require the following canonical commutation relations
\begin{equation}
\label{com}
\left[a_{N}^{\sharp} (q,v),a_{N}^{\sharp} (q',v') \right]=0, \text{ }  \text{ }   \left[ a_N (q,v), a_{N}^{*}  (q',v')\right]=\frac{1}{N} \delta(q-q') \otimes \delta(v-v')
\end{equation}
The annihilation- and creation operators, $a_N$ and $a^{*}_{N}$, are operator-valued distributions acting on \textit{Fock space}, $\mathcal{F}$, a Hilbert space defined by
\begin{equation}
\label{Fock}
\mathcal{F}:=\bigoplus_{n=0}^{\infty} \mathcal{F}^{(n)},
\end{equation}
where the n-particle subspace is given by
$$\mathcal{F}^{(n)}\simeq L^{2}(\mathbb{R}^{6n})_{sym},$$
the space of square-integrable functions on $\mathbb{R}^{6n}$ symmetric under arbitrary exchanges of arguments $(q_i,v_i)$ and $(q_j,v_j)$, $i,j=1,...,n$, and $\mathcal{F}^{0}:=\mathbb{C}\mid0\rangle$, where $\mid0\rangle$ is the \textit{vacuum vector} in $\mathcal{F}$, with $a_{N}(q,v)\mid0\rangle \equiv 0$ and $\Vert\mid0 \rangle\Vert = 1$. The quantization of Vlasov theory is supposed to inherit the local gauge symmetry of the original theory, which acts on annihilation- and creation operators by
\begin{equation}
\label{qgauge}
a^{\sharp}_{N}(q,v) \mapsto e^{\pm i\phi(q,v)} a^{\sharp}_{N}(q,v),
\end{equation}
$\phi(q,v)\in\mathbb{R}$, $a_N^{\sharp}= a_N$ or $a_N^{*}$. Thus, the ``observables'' of the quantized Vlasov theory are given by operators of the form
\begin{equation}
\label{qobs}
\hat{A}(a_{N}, a_{N}^{*};\underline{w})=\sum_{n=1}^{\infty} \int w_n(\underline{q}_n,\underline{v}_n) :\prod_{i=1}^{n} a^{*}_{N}(q_i,v_i) a_{N}(q_i,v_i) d^{3}{q}_i d^{3}{v}_i: 
\end{equation}
where $:(\cdot):$ denotes the usual Wick ordering, (all $a^{*}_{N}$'s to the left of all $a_{N}$'s). These operators only depend on the particle density $a^{*}_{N}(q,v)a_{N}(q,v)$, $(q,v)\in\mathbb{R}^{6}$, but are $independent$ of the choice of phases of the annihilation- and creation operators. They are therefore invariant under the gauge transformations (\ref{qgauge}). It is not difficult to verify that, as a consequence of gauge invariance, all the ``observables'' $\hat{A}$ commute with each other. Thus, they generate an $abelian$ algebra, $\mathcal{A}^{Q}$, of operators on $\mathcal{F}$.\\
Next, we study the dynamics of the quantization of Vlasov theory. Vectors, $\Psi$, in Fock space $\mathcal{F}$ are sequences,
$$\Psi=\lbrace\Psi^{(n)}\rbrace_{n=0}^{\infty},$$
where 
$$ \Psi^{(n)}=\int\psi^{(n)}(\underline{q}_{n},\underline{v}_{n})\prod_{i=1}^{n}a^{*}_{N}(q_i,v_i)d^{3}q_id^{3}v_i\mid0\rangle$$
with $\psi^{n}\in L^{2}(\mathbb{R}^{6n})_{sym}$. 
We propose to describe the time evolution of such state vectors in the Schr\"odinger picture. The Schr\"odinger equation takes the form

\begin{equation}
\label{Schrodinger}
\frac{i}{N} \frac{\partial \Psi(t)}{\partial t} =\hat{H}_N  \Psi(t), \Psi(t)\in \mathcal{F},
\end{equation}
where the Hamiltonian, $\hat{H}_N$, of the quantized Vlasov theory is obtained by substituting $\alpha^{\sharp}$ by $a^{\sharp}_N$ in the expression for the classical Hamilton functional and, hence, is given by

\begin{equation}
\begin{split}
\label{qHamiltonian}
\hat{H}_N:= i \int d^{3}q d^{3}v \text{ }& a^{*}_{N} (q,v) \left[ \frac{}{}  -v \cdot \nabla_q + \nabla_q V \cdot \nabla_v  + \right.\\
&\left. \left( \int d^{3}q' d^{3}v' \nabla W (q-q') a^{*}_{N} (q',v') a_{N}  (q',v') \right) \cdot \nabla_v \right] a_{N} (q,v)
\end{split}
\end{equation}
We note that the n-particle subspaces $\mathcal{F}^{(n)}$ are invariant under the dynamics generated by the Hamiltonian $\hat{H}_N$.
It is quite straightforward to show that the n-particle densities
$$f_n(t):=\vert \psi^{(n)} (t) \vert^{2}$$
are solutions of the \textit{Liouville equation}
$$ \partial_{t}f_n(t)= \lbrace H_{N}^{n}, f_n(t)\rbrace,$$
where the Poisson bracket corresponds to the symplectic 2-form 
$\sigma_n= \sum_{i=1}^{n} dq^{i} \wedge dv^{i}$ 
on n-particle phase space $\mathbb{R}^{6n}$, and the n-particle Hamilton function, $H_{N}^{n}(\underline{q},\underline{v})$, is given by 
$$H_{N}^{n}(\underline{q},\underline{v})=\sum_{i=1}^{n} \left(  \frac{v_{i}^{2}}{2} +V(q_i) \right) +\frac{1}{N} \sum_{i < j}  W(q_i-q_j). $$
 Thus, somewhat surprisingly,  $f_n(t)=\vert\psi^{n}(t)\vert^{2}$  turns out to be the phase-space probability density at time $t$ of a system of $n$ identical interacting point-particles evolving in time according to Newton's equations of motion, for every $n=1,2,...$ In other words, the Newtonian mechanics of an arbitrary number of identical point-particles may be understood as the ``quantization'' of Vlasov theory. However, because the algebra $\mathcal{A}^{Q}$ is $abelian$, Newtonian mechanics is $realistic$ and $deterministic$.

There is a more general way of preserving invariance of physical quantities (``observables'') under the gauge transformations (\ref{gauge}) and (\ref{qgauge}) than the one taken above. It would enable us to generalize our notion of ``observables'', generating a $non-abelian$ algebra $\mathcal{\bar{A}}^{Q}\supset\mathcal{A}^{Q}$. It relies on interpreting $\alpha$ as a section of a complex line bundle and introducing a connection, $\nabla$, on this bundle that gives rise to a notion of parallel transport. Using parallel transport, one can introduce gauge-invariant operators that do not only depend on $f(q,v)=\vert\alpha(q,v)\vert^{2}$ and that, upon quantization, generate a non-commutative algebra $\bar{\mathcal{A}}^{Q}$. This story is related to ``pre-quantization'', and we won't go into it here.\\

The passage from Vlasov- to wave mechanics, in the sense Schr\"odinger originally understood his theory, can be viewed as arising from a ``deformation'' of the factorization of densities $f$ into a product $\bar{\alpha}\cdot{\alpha}$.
 Instead of factorizing $f$ in this way, we may view $f$ as the \textit{Wigner transform} of a wave function, $\psi$, on one-particle $configuration$ space $\mathbb{R}^{3}$. Thus we consider functions, $f^{\hbar}$, on one-particle $phase$ space given by
 \begin{equation}
 \label{Wignertr}
f^{\hbar}(q,v)=\frac{1}{(2 \pi)^{3}}\int d^{3}q' e^{iq'\cdot v} \overline{\psi\left(q- \frac{\hbar q'}{2} \right)}\psi\left(q+ \frac{\hbar q'}{2} \right).
\end{equation}
The function $f^{\hbar}$ is called the Wigner transform of $\psi$. (Unfortunately, though, the Wigner transform of a wave function need not be non-negative.)
Next, let $f^{\hbar}_t$ be the Wigner transform of a wave function $\psi_t$, $t\in\mathbb{R}$, 
where $\psi_t$ solves the Hartree (non-linear Schr\"odinger) equation
\begin{equation}
\label{Hartree} 
i \hbar \partial_{t} \psi_t= \left[ -\frac{\hbar^{2}}{2} \Delta +V \right] \psi_t + \left[ \mid \psi_t \mid^{2} *W \right] \psi_t
\end{equation}
Here $\Delta$ is the Laplacian and the potentials $V$ and $W$ are as above; (the mass of the particles is set to 1).
It has been shown in \cite{NS} that $f^{\hbar}_t$ approaches a solution of the Vlasov-Poisson equation, as $\hbar$ tends to 0, i.e., Vlasov theory can be recovered as the classical limit of Hartree theory.\\ 
It is well known that the Hartree equations for $(\psi,\bar{\psi})$ can be interpreted as the Hamiltonian equations of motion of a Hamiltonian system with infinitely many degrees of freedom, the Hamilton functional corresponding to the energy functional of the Hartree equation. The pairs $(\psi,\bar{\psi})$ may be interpreted as complex coordinates of an infinite-dimensional complex phase space, $\Gamma_H$.
General observables, $A(\psi,\bar{\psi})$, are functionals on $\Gamma_H$ that must be invariant under $global$ gauge transformations,
$$\psi^{\sharp}(q) \mapsto e^{\pm i \phi} \psi^{\sharp}(q), \text{ }\phi\in\mathbb{R},$$
and, thus, must have the form
\begin{equation}
\label{obse}
A(\psi,\bar{\psi};\underline{w}) = \sum_{n=1}^{\infty} \int \left(\prod_{i=1}^{n} \overline{\psi(q_i)} d^{3}q_i\right)w_{n} (\underline{q}_{n},\underline{q'}_{n}) \prod_{j=1}^{n} \psi(q'_j)d^{3}q'_j,
\end{equation}
where the kernels $w_n$ are smooth and of rapid decay at $\infty$ and become ``small'', as $n\rightarrow\infty$. These observables generate an $abelian$ algebra, and the resulting theory is realistic and deterministic. If $W$ vanishes Eq.(\ref{Hartree}) is the usual time-dependent Schr\"odinger equation. Its interpretation as a deterministic Hamiltonian system is the interpretation Schr\"odinger initially wanted to give to his wave mechanics.

Quantization of Hartree theory can be carried out by following Dirac's recipe, in a way very similar to what we have explained above on the example of Vlasov theory. Wave functions $\psi$ are replaced by annihilation operators, $\hat{\psi}_{N}$, and their complex conjugates by creation operators, $\hat{\psi}^{*}_{N}$, satisfying the canonical commutation relations. (Details can be found in  \cite{FRKP}.) The quantum theory so obtained describes a gas of $n=0,1,..$ bosons of mass $m=1$ with two-body interactions given by the potential $\frac{1}{N}W$ and under the influence of an external potential $V$, in the formalism of ``second quantization''. The operators, $\hat{A}$, corresponding to the quantization of the functions $A$ defined in (\ref{obse}) do $not$, in general, commute, because we have only required $global$ gauge invariance of physical quantities. Thus, if the system $S$ is a Bose gas described by the quantization of Hartree theory the algebra $\mathcal{A}_S$ of physical quantities pertaining to $S$ is non-abelian.\\
By coupling $\psi^{\sharp}$ ($\hat{\psi}^{\sharp}_N$) to a U(1)-connection (interpreted, e.g., as an electromagnetic vector potential) one can promote the global gauge symmetry of Hartree theory and of its quantized version to a $local$ gauge symmetry. But the algebra generated by physical quantities of the quantized Hartree theory (i.e., a theory of charged Bose gases) remains non-abelian.\\
In conclusion, the quantization of Hartree theory, viewed as an infinite-dimensional Hamiltonian system, naturally leads to the quantum mechanics of interacting Bose gases, expressed in the formalism of ``second quantization''. Conversely, Hartree theory can be obtained as the mean-field (weak-coupling, or ``classical'') limit of the theory of interacting Bose gases; see \cite{Hepp}, \cite{FRKP'} and references given there.

Our discussion can be usefully summarized in the following diagram.
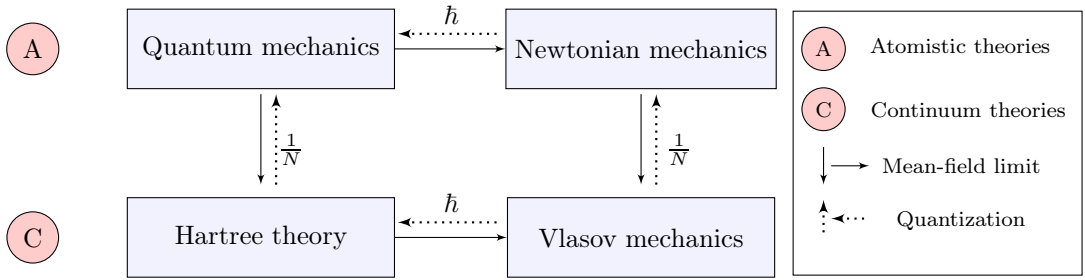
\tikzstyle{block} = [draw, fill=blue!5, rectangle, 
    minimum height=3em, minimum width=10em]
\tikzstyle{sum} = [draw, fill=red!20, circle, node distance=1cm]
\tikzstyle{input} = [coordinate]
\tikzstyle{output} = [coordinate]
\tikzstyle{pinstyle} = [pin edge={to-,thin,black}]

\begin{center}
\begin{tikzpicture}[auto, node distance=2cm,>=latex']
    \node [sum] (sum) {A};
    \node [block, right of=sum, node distance=3cm] (quantum) {Quantum mechanics};
    \node [block, right of=quantum,
            node distance=5cm] (newton) {Newtonian mechanics};
    \node [block, below of=newton, node distance=2.5cm] (vlasov) {Vlasov mechanics};
     \node [block, left of=vlasov, node distance=5cm] (hartree) {Hartree theory};
     \node [sum, left of=hartree, node distance=3cm] (sum2) {C};

  \draw [->] (quantum) -- (newton);
 \draw [->] (hartree) --  (vlasov);
\draw[->] (3,-0.6) -- (3,-1.8);
\draw[thick,dotted,->] (3.2,-1.8) -- (3.2,-0.6);
\draw[->] (8,-0.6) -- (8,-1.8);
\draw[thick,dotted,->] (8.2,-1.8) -- (8.2,-0.6);

\draw[thick,dotted,->] (6.1,0.2) -- (4.8,0.2);
\draw[thick,dotted,->] (6.1,-2.3) -- (4.8,-2.3);

\draw (10,-3) rectangle (13.8,0.5); 
\draw (3.4,-1) node[below]{$\frac{1}{N}$};
\draw (8.5,-1) node[below]{$\frac{1}{N}$};

\draw (5.5,0.7) node[below]{$\hbar$};
\draw (5.5,-1.8) node[below]{$\hbar$};

\begin{footnotesize}
 \node [sum,right of=newton,node distance=2.4cm] (sum3) {A};
      \node [sum,below of=sum3,node distance=0.8cm] (sum4) {C};
\draw (12.2,0.25) node[below]{Atomistic theories};
\draw (12.31,-0.6) node[below]{  Continuum theories};
\draw[->] (10.4,-1.3) -- (10.4,-1.8);
\draw[->] (10.5,-1.55) -- (11,-1.55);
\draw[thick,dotted,<-] (10.4,-2) -- (10.4,-2.5);
\draw[thick,dotted,<-] (10.5,-2.25) -- (11,-2.25);
\draw (12.21,-1.35) node[below]{Mean-field limit};
\draw (12.21,-2.05) node[below]{Quantization};

\end{footnotesize}

\end{tikzpicture}
\end{center}

\section{Probabilities of histories, dephasing and decoherence}
This section may well be the most important one in these notes. We clarify the notion of ``possible event'' in a system $S$; we introduce the notion of ``interference'' between a possible event and its complement, which refers to a phenomenon typical of quantum theory; related to interference, we must explain what is meant by the ``evidence that a possible event in $S$ can materialize'', given the future.

Most importantly, we have to explain how the empirical probabilities of (time-ordered) sequences of possible events in a system \mbox{-- histories --} can be calculated if we know the state of the system, $S$, composed with its environment, $E$.
Logically, this would oblige us to present an outline of the theory of preparation of states of quantum-mechanical systems, which, however, goes somewhat beyond the scope of these notes and will be described elsewhere; (but see \cite{FGrS, DeRoeck} for recent results relevant in this context). We will conclude this section with a brief outline of the role of ``dephasing'' and ``decoherence'' (which are properties of the time evolution of $S\vee{E}$) in suppressing interference terms between possible events and their complements, thus rendering them complementary, i.e., mutually exclusive, in the classical sense.\\
In this section, the system $E$ represents either the  ``environment'' of the system $S$ or the  ``equipment'' used  by an observer to carry out measurements of physical quantities pertaining to $S$. Equipment is a particular type of environment that can be controlled (to some extent) by an observer in  order to measure a specific physical quantity, in the sense that its state and its interactions with $S$ can be tuned by the observer. (For a generic environment, this is impossible.)
\subsection{Quantum probabilities}
A general recipe for how to calculate empirical probabilities, or ``frequencies'', of time-ordered sequences of possible events (histories) in a system described by a quantum theory has been proposed by L\"uders, Schwinger and Wigner; see \cite{Luders}, \cite{Schwinger},\cite{Wigner}. To describe their recipe, we consider a system, $S$, coupled to an environment, $E$, that is supposed to trigger events, e.g., values of some family of physical quantities, $a^j \in \mathcal{A}_S$, $j=1,...,n,$  measured at times $t_1<...<t_n$. We denote by $\mathcal{B}_{S\vee{E}}$ the algebra of possible events in $S\vee{E}$ and by $a_{t} := \alpha_{t}(a)$ the operator in $\mathcal{B}_{S\vee{E}}$ corresponding to the physical quantity $a\in\mathcal{A}_{S}$ at time $t$, where $\alpha_{t}(\cdot)$ is the time-translation automorphism on $\mathcal{B}_{E\vee{S}}\supset\mathcal{A}_{S}$. Possible events are represented by spectral projections $P_{a^{j}_{t_j}}(I_j)$, where $I_{j} \subset \sigma(a^j)$, the spectrum of $a^j$. 
A history is a sequence, $\lbrace P_n,...,P_1 \rbrace$, of possible events
$$P_j=P_{a^{j}_{t_j}} (I_j),$$
where $a^{j}=(a^{j})^{*} \in \mathcal{A}_S$, $I_j \subset \sigma (a^{j})$, and $t_1<...<t_n$. 
Quantum mechanics predicts the probability, or $``frequency"$, $\mathcal{F}_{\omega}\lbrace P_n,...,P_1 \rbrace$, of  a history $\lbrace P_n,...,P_1 \rbrace$ of possible events to be observed in actual experiments, given a state $\omega$ of $S \vee E$.\\\\

$\bf{``Master}$  $\bf{Formula}$'' for frequencies:
\begin{equation}
\label{freq}
\mathcal{F}_{\omega} \lbrace P_n,...,P_1 \rbrace := \omega (P_1 ... P_{n-1} P_n P_{n-1}...P_1)
\end{equation}

Some properties of frequencies.
\begin{enumerate}[(i)]
\item Since $P_1 .... P_{n-1} P_{n} P_{n-1} ... P_1 = Q Q^{*}\geq 0$, with $Q=P_1 ... P_{n-1} P_{n}$, and $\omega$ is a state, we conclude that $$\mathcal{F}_{\omega} \lbrace P_n,...,P_1 \rbrace \geq 0$$
\item Moreover, since $\omega$ is normalized and because $\Vert QQ^{*} \Vert \leq 1$, we have that
 $$\mathcal{F}_{\omega} \lbrace P_n,...,P_1 \rbrace \leq 1$$
Thus, $\mathcal{F}_{\omega} \lbrace P_n,...,P_1 \rbrace$ can be interpreted as a probability.\\
\item We let $P_{j}^{k}$, $k=1,...,K_j$, with $P_{j}^{1}=:P_{j}$, denote all possible events that may be observed in a measurement of the quantity $a_{t_j}^{j}$, using the equipment described by $E$. Then
\begin{equation}
\sum_{k=1}^{K_j}P_{j}^{k}=\mathbb{I}
\end{equation}
We define
$$\sigma_{n}:= \lbrace (k_1,...,k_n) \mid k_j = 1,...,K_j, \text{ for } j=1,...,n \rbrace$$
It is easy to see that
\begin{equation}
\label{probsr}
\sum_{(k_1,...,k_n) \in \sigma_n} \mathcal{F}_{\omega} \lbrace P_{n}^{k_n},...,P_{1}^{k_1} \rbrace=1
\end{equation}
Thus, $\lbrace \mathcal{F}_{\omega} \lbrace P_{n}^{k_n},...,P_{1}^{k_1} \rbrace \rbrace$ defines a probability measure on the set $\sigma_n$.
\item Note, however, that, in general, 
\begin{equation}
\label{sumrule}
\sum_{k=1}^{K_j}\mathcal{F}_{\omega}\lbrace P_n,...,P_{j}^{k},...,P_1 \rbrace \neq \mathcal{F}_{\omega} \lbrace P_n,...,P_{j+1},P_{j-1},...,P_1 \rbrace,
\end{equation}
for $j<n$,
because of quantum-mechanical interferences. If $K_j>2$ this renders a consistent definition of the conditional probability of the event $P_j$, given $P_n,...,P_{j+1},P_{j-1},...,P_1$, impossible. This observation is vaguely related to the Kochen-Specker theorem. It will be discussed in more detail, below.
\end{enumerate}
Property (\ref{sumrule}) points to the most characteristic difference between quantum probabilities and the probabilities appearing in realistic theories and can be interpreted as saying that quantum theories are not ``realistic'' theories.

\subsection{Indeterminism in quantum theory}
Recall that if we were considering a realistic theory and if $\omega$ were a pure state, i.e., a Dirac measure on the spectrum of $\mathcal{B}_{S\vee{E}}$, and, thus, a pure state on $\mathcal{A}_{S}$, then $ \mathcal{F}_{\omega} \lbrace P_n,...,P_1 \rbrace =0 \text{ or } 1$. The frequencies of a quantum theory do not, in general, obey such 0--1 laws!\\
We choose a pure state $\omega$ on $\mathcal{B}_{S\vee E}$, whose restriction to $\mathcal{A}_{S}$ may be assumed to be pure, too, meaning that $S$ and $E$ are not entangled in this state. Let $\mathcal{H}_{\omega}$ denote the Hilbert space obtained from $(\mathcal{B}_{S\vee{E}}, \omega)$ by the GNS construction, and let $P_{\omega}$ denote the orthogonal projection onto the cyclic vector $\xi_{\omega}\in \mathcal{H}_{\omega}$ corresponding to the state $\omega$. We consider a history, $\lbrace P_n,...,P_1 \rbrace$, with the property that the projections $P_1$ and $P_{\omega}$ do not commute with each other, i.e., the vectors $P_{1}\xi_{\omega}$ and $P_{1}^{\perp}\xi_{\omega}$ are both different from 0, with $P_{1}^{\perp}:=\mathbb{I}-P_1$, and $Q^{*}\xi_{\omega}\neq 0$, where $Q=P_1...P_n$. Then

\begin{equation}
0< (\xi_{\omega}, QQ^{*}\xi_{\omega}) = \mathcal{F}_{\omega} \lbrace P_n,...,P_1 \rbrace < 1
\end{equation}
Thus, quantum-mechanical frequencies do not, in general, obey 0--1 laws, even if the state $\omega$ is pure (and even if its restriction to $\mathcal{A}_{S}$ is pure, too). We conclude that quantum mechanics is $non{-}deterministic$. It deserves to be mentioned that, typically, a pure state $\omega$ on $\mathcal{B}_{S\vee{E}}$ does not determine a pure state on $\mathcal{A}_{S}$ because of entanglement and that, even if the restriction of $\omega$ to $\mathcal{A}_{S}$ were pure, time evolution will usually cause entanglement between $S$ and $E$.\\
In order for the argument just outlined to be conclusive a discussion of the ``theory of preparation of states'' of quantum-mechanical systems would be called for. A more detailed analysis of the subject of this section will appear elsewhere.

\subsection{Interferences and ``$\delta-$consistent histories''}
We recall that, in a realistic theory, possible events are characteristic functions of measurable subsets of some measure space ($\Omega, \sigma$) and states are given by probability measures on $\Omega$. If $\Delta^{1}, \Delta^{2},..., \Delta^{K}$, with $\Delta^i \cap \Delta^j=\emptyset$, for $i \neq j$,  $\Delta:=\bigcup_{i=1}^{K} \Delta^{i}$, and $\Sigma$ are measurable subsets of $\Omega$ and if $\mu$ is a probability measure on $\Omega$ then  
$$\sum_{i=1}^{K}\mu (\Sigma \cap \Delta^{i})=\mu(\Sigma \cap \Delta)$$ 
In contrast, in a quantum theory, $interferences$ between possible events corresponding to mutually orthogonal spectral projections of a physical quantity may arise, given their (past and) future and a state of $S\vee E$. If $P_{j}^{1},...,P_{j}^{K}$ are such projections, with $P_{j}^{(K)}=\sum_{i=1}^{K} P_{j}^{i}$ then, in general,
\begin{equation}
\label{eq-int}
\sum_{i=1}^{K}\mathcal{F}_{\omega} \lbrace P_n,...,P_{j}^{i},...,P_1 \rbrace \neq \mathcal{F}_{\omega} \lbrace P_n,...,P_{j}^{(K)},...,P_1 \rbrace,
\end{equation}
unless $j=n$, because $P_{j}^{i}$ does not necessarily commute with $P_{j+1}...P_n...P_{j+1}$, and hence interference terms
$$Re \left(\omega(P_1 ...P_{j}^{i} ... P_n ...P_{j}^{k}...P_1) \right)$$ may be non-zero, for $i\neq k$. 
In particular, in general
$$\mathcal{F}_{\omega} \lbrace P_n,...,P_j,...,P_1 \rbrace + \mathcal{F}_{\omega} \lbrace P_n,...,P_{j}^{\perp},..., P_1 \rbrace \neq \mathcal{F}_{\omega} \lbrace P_n,..., P_{j+1},P_{j-1},..., P_1 \rbrace,$$
where $P^{\perp}=\mathbb{I} - P$; see Eq. (\ref{sumrule}).\\
Here are some important consequences of non-vanishing interference terms: If, in a measurement of a physical quantity, $a_j$, at time $t_{j}$, there are more than two possible measurement outcomes, $P_j=:P_{j}^{1}$,...,$P_{j}^{K_j}$, with $K_{j}\geq3$, then it is impossible to unambiguously define the ``conditional probability'' of $P_j$, given $P_n,...,P_{j+1}, P_{j-1},...,P_1$. One might be tempted to define the conditional probability, $\mathcal{F}_{\omega} \lbrace P_n,...\mid P_j \mid...,P_1 \rbrace$, of $P_j$, given $P_n,...,P_{j+1}, P_{j-1},..., P_1$, by the formula
\begin{equation}
\label{condprob}
\mathcal{F}_{\omega} \lbrace P_n,...\mid P_j \mid...,P_1 \rbrace := \frac{\mathcal{F}_{\omega}\lbrace P_n,...,P_j,..., P_1 \rbrace}{\mathcal{F}_{\omega}\lbrace P_n,...,P_j,..., P_1 \rbrace + \mathcal{F}_{\omega}\lbrace P_n,...,P_{j}^{\perp},..., P_1 \rbrace}
\end{equation}
However, since in general
$$\mathcal{F}_{\omega}\lbrace P_n,...,P_{j}^{\perp},..., P_1 \rbrace \neq \sum_{i=2}^{K_j} \mathcal{F}_{\omega}\lbrace P_n,...,P_{j}^{i},..., P_1 \rbrace,$$ because of non-vanishing interference terms, the definition of 
$\mathcal{F}_{\omega} \lbrace P_n,...\mid P_j \mid...,P_1 \rbrace$ proposed above may differ from alternative definitions given by
$$ \frac{\mathcal{F}_{\omega}\lbrace P_n,...,P_j,..., P_1 \rbrace}{\sum_{i=1}^{K}\mathcal{F}_{\omega}\lbrace P_n,...,P_{j}^{i},..., P_1 \rbrace} , \text{ } \text{or}  \text{  }  \frac{\mathcal{F}_{\omega}\lbrace P_n,...,P_j,..., P_1 \rbrace}{\mathcal{F}_{\omega}\lbrace P_n,...,P_{j+1}, P_{j-1}..., P_1 \rbrace},$$
and the second object above is not even necessarily bounded above by $1$.
Only in cases where $K_{j}=2$ definition (\ref{condprob}) is meaningful.

Thus, in general, there is no meaningful notion of ``conditional probability'' of a possible event $P_j$, given its (past and) future, i.e., the (conditional) probability of observing $P_j$, given that $P_1,..., P_{j-1}, P_{j+1},..., P_n$ are observed, cannot be predicted unambiguously. 
The reason is that, in quantum theory, a possible event $P_j$ and its complement $P_{j}^{\perp}$ are, in general, not complementary (mutually exclusive) in the classical sense of this expression. Whether they are complementary or not depends on the choice of the experimental equipment, described by $E$, used to measure the observable quantity $a_j$.
If all interference terms between different possible outcomes in a measurement of $a_j$ at time $t_j$ very nearly vanish, given future measurements, (one then speaks of ``dephasing'' or ``decoherence'') then a possible event $P_j$ and its complement $P_{j}^{\perp}$ are complementary in the classical sense;  and hence $P_j$ may correspond to a ``fact'' in an actual experiment. The mechanisms of ``dephasing'' and ``decoherence'' are briefly described in Subsection 4.6.

What we are trying to convey here is well known from the analysis of concrete examples, such as the \textit{double-slit experiment}. 
In this experiment, the projection $P_2$ may represent the possible event that an electron, after having passed a shield with two slits, reaches a region $\Delta$ of a screen, where it triggers the emission of a flash of light. The projection $P_{1}^{r}$ represents the event that the electron has passed through the slit on the right of the shield, while $P_{1}^{l} = (P_{1}^{r})^{\perp}$ stands for the possible event that the electron has passed through the slit on the left of the shield.  Due to usually non-vanishing interference terms, Re $\omega(P_{1}^{r} P_2 P_{1}^{l})$, 
$$\mathcal{F}_{\omega} \lbrace P_2,P_{1}^{r} \rbrace +\mathcal{F}_{\omega} \lbrace P_2,P_{1}^{l}\rbrace \neq \mathcal{F}_{\omega} \lbrace P_2 \rbrace $$
This can be tested, experimentally, because $\mathcal{F}_{\omega} \lbrace P_2,P_{1}^{r} \rbrace $ can be determined from experiments where the slit on the left of the shield is blocked, while $\mathcal{F}_{\omega} \lbrace P_2,P_{1}^{l} \rbrace $ can be determined from experiments where the slit on the right is blocked.
Finally, $ \mathcal{F}_{\omega} \lbrace P_2 \rbrace$ can be determined from experiments where both slits are left open.\\

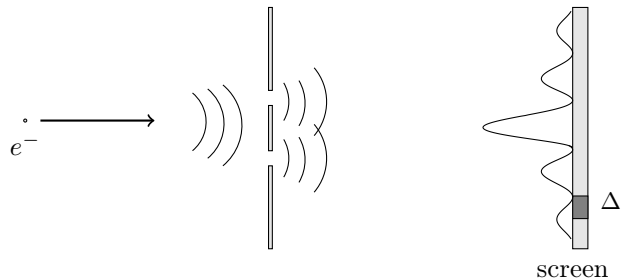
\begin{center}
\begin{tikzpicture}   
\draw[thick,->] (1,-1) -- (2.5,-1);
\draw[fill=black!10] (8,-2.7) rectangle (8.2,0.5); 
\draw (8,-2.8) node[below]{screen};
\draw[fill=black!10] (4,-0.6) rectangle (4.05,0.5); 
\draw[fill=black!10] (4,-1.4) rectangle (4.05,-0.8); 
\draw[fill=black!10] (4,-2.7) rectangle (4.05,-1.6); 
\draw[fill=black!10] (0.8,-1) circle (0.02);
\draw (0.8,-1.05) node[below]{$e^{-}$};

\draw (8.5,-2.3) node[above]{\small $\Delta$};
\draw[fill=black!50] (8,-2.3) rectangle (8.2,-2.);

\draw (3,-1.4) arc (-50:50:0.5) ;
\draw (3.2,-1.5) arc (-50:50:0.6) ;
\draw (3.4,-1.6) arc (-50:50:0.7) ;

\draw (4.2,-1) arc (-30:30:0.5) ;
\draw (4.4,-1.07) arc (-30:30:0.6) ;
\draw (4.6,-1.2) arc (-40:40:0.7) ;

\draw (4.2,-1.75) arc (-30:30:0.5) ;
\draw (4.4,-1.82) arc (-30:30:0.6) ;
\draw (4.6,-1.95) arc (-40:40:0.7) ;


\begin{scope}[xshift=8.0cm,yshift=-1.99cm,rotate=90]
\draw plot [domain=0.63:2.44,samples=200] (\x,{sin(5*\x r)^2/\x^2});
\end{scope}

\begin{scope}[xshift=8.0cm,yshift=-0.13cm,rotate=90]
\draw plot [domain=-2.44:-1.2,samples=200] (\x,{-(sin(5*\x r))^2/\x^2});
\end{scope}

\end{tikzpicture}
\end{center}

We now imagine that a laser lamp, emitting light of a wave length much smaller than the distance between the two slits in the shield, is turned on in the cavity between the shield and the screen. We then expect that the interference pattern, observed on the screen when both slits in the shield are open and the laser lamp is turned off, gradually disappears when the laser lamp is turned on and its intensity is increased. This is due to scattering processes between the electron and the photons in the laser beam, which serve to track the trajectory of the electron. 

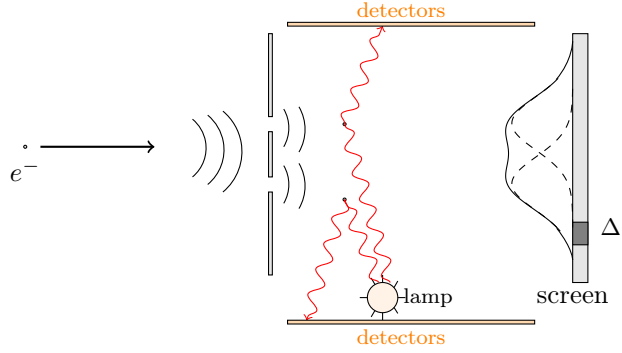
\begin{center}
\begin{tikzpicture}
   
\draw[thick,->] (1,-1) -- (2.5,-1);
\draw[fill=black!10] (8,-2.8) rectangle (8.2,0.5); 
\draw (8,-2.8) node[below]{screen};
\draw[fill=black!10] (4,-0.6) rectangle (4.05,0.5); 
\draw[fill=black!10] (4,-1.4) rectangle (4.05,-0.8); 
\draw[fill=black!10] (4,-2.7) rectangle (4.05,-1.6); 
\draw[fill=black!10] (0.8,-1) circle (0.02);
\draw (0.8,-1.05) node[below]{$e^{-}$};

\draw[fill=orange!10] (5.5,-3) circle (0.2);
\draw[-] (5.7,-3) -- (5.8,-3);
\draw[-] (5.65,-2.86771) -- (5.7,-2.78);
\draw[-] (5.5,-2.8) -- (5.5,-2.7);
\draw[-] (5.4,-2.82679) -- (5.35,-2.74);
\draw[-] (5.3,-3) -- (5.2,-3);

\draw[-] (5.65,-3.13) -- (5.7,-3.21);
\draw[-] (5.5,-3.2) -- (5.5,-3.3);
\draw[-] (5.4,-3.17) -- (5.35,-3.256);
\draw (5.65,-3) node[right]{\footnotesize  lamp};

\draw[fill=black!10] (5,-1.7) circle (0.02);
\draw[fill=black!10] (5,-0.7) circle (0.02);

\draw [-,snake=snake,color=red] (5.45,-2.8) -- (5,-1.7);
\draw [-,snake=snake,color=red] (5.6,-2.8) -- (5,-0.7);
\draw [->,snake=snake,color=red] (5,-1.7) -- (4.5,-3.3);
\draw [->,snake=snake,color=red] (5,-0.7) -- (5.5,0.6);

\draw (8.5,-2.3) node[above]{\small $\Delta$};
\draw[fill=black!50] (8,-2.3) rectangle (8.2,-2.); 

\draw (3,-1.4) arc (-50:50:0.5) ;
\draw (3.2,-1.5) arc (-50:50:0.6) ;
\draw (3.4,-1.6) arc (-50:50:0.7) ;

\draw (4.2,-1) arc (-30:30:0.5) ;
\draw (4.4,-1.07) arc (-30:30:0.6) ;

\draw (4.2,-1.75) arc (-30:30:0.5) ;
\draw (4.4,-1.82) arc (-30:30:0.6) ;


\begin{scope}[xshift=8.cm,yshift=-1.99cm,rotate=90]
\draw plot [domain=-0.13:1.9,samples=200] (\x,{0.8*exp(-4*(\x-1.3)^2)+0.8*exp(-4*(\x-0.5)^2)});
\end{scope}

\begin{scope}[xshift=8.cm,yshift=-1.99cm,rotate=90]
\draw [dashed] plot [domain=-0.08:1.8,samples=100] (\x,{0.8*exp(-4*(\x-1.3)^2)});
\draw   [dashed]   plot [domain=-0.08:1.8,samples=100] (\x,{0.8*exp(-4*(\x-0.5)^2)});
\end{scope}

\draw (7.782,-0.11) .. controls (8,0.15) and (8,0.4) .. (8,0.4);
\draw (7.79,-2.08) .. controls (8,-2.35) and (8,-2.5) .. (8,-2.5);

\draw[fill=orange!30] (4.25,0.6) rectangle (7.5,0.65); 
\draw[fill=orange!30] (4.25,-3.3) rectangle (7.5,-3.35); 
\draw[color=orange] (5.75,-3.3) node[below]{\footnotesize detectors};
\draw[color=orange] (5.75,0.6) node[above]{\footnotesize detectors};

\end{tikzpicture}
\end{center}
If the electromagnetic field emitted by the laser is included in the theoretical description of the equipment, $E$, used in this experiment then the disappearance of the interference pattern on the screen can be understood as the result of decoherence, which makes the interference term Re $\omega(P_{1}^{r} P_2 P_{1}^{l})$ tend to 0, as the wave length of the laser decreases and its intensity is cranked up, and, hence, renders the possible events $P_{1}^{r}$ and $P_{1}^{l}$ complementary in the classical sense.\\
The experiment described here has first been proposed by Feynman \cite{Fe}. A theoretical analysis has been given, e.g., in \cite{Schil}.\\

Inspired by this example, we introduce a notion of ``$\delta-$consistent histories''.\\

\begin{definition}
$\delta-$consistent histories
\end{definition}

Let $\lbrace P_n,...,P_j,...,P_1 \rbrace$ be a history of a system $S$ and $\omega$ the state of the system $S\vee{E}$. (Without loss of generality, we suppose that $P_j=:P_{j}^{1}$, for $j=1,...,n$.)
We define a quantity, $\mathcal{E}_{\omega}^{(j)}$ -- called the ``evidence for one of the possible events 
$\lbrace P_{j}^{k} \rbrace_{k=1}^{K_j}$ to materialize, (given past and future events and the state of the system)'' -- by
$$\mathcal{E}_{\omega}^{(j)}:=1- \sum_{1 \leq k,l \leq K_j, k\neq l} \vert \omega(P_1 ... P_{j-1} P_j^{k} P_{j+1}  .... P_{n} ...P_{j+1}  P_j^{l} P_{j-1}...P_1) \vert,$$
$j=1,...,n-1$.
We note that if $\mathcal{E}_{\omega}^{(j)}$ were $=1$ (i.e., if the interference terms appearing in its definition vanished) then the possible events $P_{j}^{1},...,P_{j}^{K_j}$ would be mutually exclusive (in the classical sense), and hence one of them would $have$ to happen, given the state of $S\vee E$ and (past and) future possible events. If the value of $\mathcal{E}_{\omega}^{(j)}$ is so close to 1 that it cannot be distinguished from 1 then the system responds to a measurement of the physical  quantity $a_{t_j}^{j}$,  as if precisely one of the possible events $\lbrace P_{j}^{k} \rbrace_{k=1}^{K_j}$ happened. If the value of $\mathcal{E}_{\omega}^{(j)}$ is appreciably smaller than 1 then it does not make sense to say that one of the possible events $\lbrace P_{j}^{k} \rbrace_{k=1}^{K_j}$ happen, given future measurements.\\
A history $\lbrace P_n,...,P_1 \rbrace$ of $S$ is said to be $\delta-$consistent with respect to the state $\omega$ of $S\vee E$ if
$$\text{min}_{j=1,...,n}\mathcal{E}_{\omega}^{(j)} \geq \delta,$$
for some $\delta\leq 1$. 

A history may correspond to a sequence of $``facts$'' if it is $\delta$-consistent, with $\delta$ very close to 1, i.e., $0\leq 1 - \delta <<1$. If $\delta=1$, we say that the history is ``$consistent$''; see \cite{Griffiths}. A consistent history is ``classical'' in the sense that, in a measurement of a physical quantity $a_j$ at time $t_j$, the event $P_j=P_{j}^{1}$ is complementary to the events $P_{j}^{k}$, $k=2,...,K_j$, in the classical sense that these events are mutually exclusive, for all $j=1,...,n$. Of course, the notion of ``consistent histories'' is an idealization in so far as, in realistic experiments, consistent histories are usually not encountered. The significance of the mechanisms of $``dephasing"$ and $``decoherence"$ (see Subsection 4.6) is that they render histories $\delta-$consistent, with $\delta>0$, and hence ``classical'', for $\delta$ very close to $1$.

In order to explain these matters in a concrete situation, we return to the double slit experiment:  Assuming that both slits in the shield are open and that the intensity of the laser is finite, the histories $\lbrace P_2,P_{1}^{r/l}\rbrace$ are never consistent, because, quantum mechanically, it is impossible to say with certainty through which slit an electron has passed. When the laser lamp is turned on they are, however, $\delta-$consistent, for some $\delta>0$. The value of $\delta$ increases, as the intensity of the laser increases, and approaches $1$, as the wave length of the laser light tends to $0$ and the intensity tends to $\infty$. For, in this limit, an experiment would determine with certainty through which slit an electron has passed.

\subsection{A remark related to the Kochen-Specker theorem}
In the previous subsection, we have seen that, because of non-vanishing interference terms, it is generally impossible to define an unambiguous notion of ``\textit{conditional probability}'' of a possible event, $P_j$, given its future, $P_{j+1},...,P_n$, for $j<n$. This observation is a reflection of the fact that quantum theories can usually not be given a realistic interpretation in terms of a hidden-variables embedding and represents a (perhaps somewhat cheap) variant of the Kochen-Specker theorem. 
We fist consider a system $S$ and experimental equipment $E$ enabling one to measure physical quantities, $a_1,...,a_n$, pertaining to $S$, with the property that a measurement of each of the quantities $a_j$ can only have two possible outcomes, $P_j$ and $P_{j}^{\perp}$, with $P_{j}+P_{j}^{\perp}=\mathbb{I}$, for all $j=1,...,n$, i.e., the quantities $a_1,..., a_n$ are all binary. Under these assumptions, one may unambiguously define the conditional probability of the possible event $P_j$, given $P_1,...,P_{j-1}, P_{j+1},..., P_n$, by
\begin{equation}
\label{binary}
\mathcal{F}_{\omega} \lbrace P_n,..,P_{j+1}  \vert P_j \vert P_{j-1},..,P_{1}  \rbrace :=\frac{\mathcal{F}_{\omega} \lbrace P_n,..., P_j,..., P_{1}  \rbrace }{\mathcal{F}_{\omega} \lbrace P_n,...,P_j,...,P_{1}  \rbrace+\mathcal{F}_{\omega} \lbrace P_n,...,P_j^{\perp},...,P_{1}   \rbrace}
\end{equation}
Thus, the quantum probabilities $\mathcal{F}_{\omega} \lbrace P_1,..., P_n \rbrace$ determine a probability measure on the discrete space ${\lbrace +,- \rbrace}^{\times n}$ that predicts conditional probabilities unambiguously. \\
However, considering a system $S$ and experimental equipment $E$ enabling one to measure physical quantities, $a_1,..., a_n$, pertaining to $S$, with the property that a measurement of the quantity $a_j$ may have $K_j\geq 3$ possible outcomes, $P_j=:P_{j}^{1},..., P_{j}^{K_j}$, for some $j<n$, we run into the problem that the conditional probability of the possible event $P_j$, given the possible events $P_1,..., P_{j-1}, P_{j+1},..., P_n$, cannot be defined unambiguously, because $P_{j}^{\perp}$ can be further decomposed into a sum, $\sum_{i=2}^{K_j} P_{j}^{i}$, and 
$$\mathcal{F}_{\omega}\lbrace P_n,..., P_{j}^{\perp},..., P_n \rbrace \neq \sum_{i=2}^{K_j} \mathcal{F}_{\omega} \lbrace P_n,..., P_{j}^{i},..., P_1 \rbrace,$$
due to non-vanishing interference terms. This means that it is usually not meaningful to imagine that the possible event $P_j$ may materialize.\\
Our argument would fail if one exclusively considered consistent histories, (which is what the ``Bohmians'' appear to accomplish by restricting what they consider to be ``physical quantities'' and ``possible events'' to a class of operators that generate an $abelian$ algebra; see \cite{Durr}).

\subsection{Consistent histories in the vicinity of $\delta$-consistent histories, for $\delta\approx 1$}
\label{4.5}
In this subsection, we present a lemma showing that, in the vicinity of a $\delta-$consistent history, with $\delta$ very close to $1$, there is a consistent history. We define a sequence, $\left( C_{n}\right)_n$, of positive numbers by
$$
C_{n}:= \left\{
    \begin{array}{ll}
        0 & \mbox{if } n=1 \\
        6(4\sum_{k=1}^{n-1}  C_{k} +1) & \text{ } \forall n \geq 2
    \end{array}
\right.
$$

\begin{lemma}
\label{lem1}
Let us suppose that $\lbrace P_n,...,P_1 \rbrace$ is a history of possible events in a system $S$ with the property that
$$\vert \vert \left[ P_j, H_j\right] \vert \vert < \epsilon,$$ where $H_j:=(\Pi_{i=j+1}^{n}P_i)(\Pi_{i=n}^{j+1}P_i)$, for some sufficiently small $\epsilon$ and all $j=1,..., n-1$. Then there exists a history ,$\lbrace \tilde{P}_n,...,\tilde{P}_1 \rbrace$, of orthogonal projections with the properties that $\vert \vert  \tilde{P}_j- P_j\vert \vert < C_{n+1-j}\epsilon$ and that
\begin{equation}
 \left[ \tilde{P}_j, \tilde{H}_j\right] =0, 
\end{equation}
where $\tilde{H}_j:=(\Pi_{j+1}^{n}\tilde{P}_i)(\Pi_{i=n}^{j+1}\tilde{P}_i)$, for all j=1,...,n. The operators $\tilde{H}_j$ are orthogonal projections.

\end{lemma}
This lemma (more precisely, some straightforward generalization of it) shows that the history $\lbrace \tilde{P}_n,..., \tilde{P}_1\rbrace$ is consistent, because all interference terms vanish. More general results of this sort will be proven elsewhere. \\
The proof of Lemma \ref{lem1} is based on the following simple 

\begin{lemma}
\label{lem2}
Let $P$ be a bounded selfadjoint operator on a Hilbert space $\mathcal{H}$, and let $0<\epsilon<\frac{1}{4}$. If $\vert \vert P^{2}-P \vert \vert <\epsilon$ then there there exists an orthogonal projection, $\hat{P}$, on $\mathcal{H}$ such that $$\vert \vert \hat{P} -P \vert \vert < 2 \epsilon.$$
\end{lemma}
We note that the operator $\hat{P}$ can be chosen to be a function of the operator $P$, so that if $\left[Q,P \right]=0$, for a given operator $Q$, then $\left[Q,\hat{P} \right]=0$.
 
Proofs of these two lemmata can be found in Appendix C.

\subsection{Dephasing and decoherence}

We consider a sequence $\lbrace P_n,...,P_1 \rbrace$ of events characterizing the actual evolution of a system, $S$, coupled to a piece of equipment, $E$, confined to a compact region, $\Lambda$, of physical space. Let $j \in \lbrace 1,...,n \rbrace$, and let $t'$ be an instant of time, with $t_{j-1 } \leq t' < t_j$, when some interaction between $E$ and $S$ is turned on, with the purpose to measure a physical quantity $a^{j}_{t_j}:=\alpha_{t_{j},t_{0}}(a^{j})$ pertaining to $S$. The measurement of $a^{j}_{t_j}$ may give rise to an event $P_j$ represented by a spectral projection of $a^{j}_{t_j}$. For simplicity, we suppose that the spectrum of the operator $a_j$ consists of a finite set of eigenvalues, $\lbrace\alpha_j^{l}\rbrace_{l=1}^{K_j}$, so that the spectral decomposition of $a_{t_j}^{j}$ is given by a finite sum 
$$a^{j}_{t_j}=\sum_{l=1}^{K_j} \alpha_{j}^{l} P_{j}^{l}$$

Let  $P_j := P_{j}^{l_0}$, for some $l_0 \in \lbrace 1,...,K_j \rbrace $. Let $\rho$ denote the state of $S \vee E$ at time $t'$, before a measurement of the quantity $a_{t_j}^{j}$ is made. The possible event $P_j$ can materialize, i.e., correspond to a  \textit{fact}, (or, put differently, $P_j$ and $\lbrace P_{j}^{l} \rbrace _{l\neq l_{0}}$ mutually exclude one another), in a measurement of $a_{t_j}^{j}$, given the state $\rho$ and future possible events  $P_{j+1}, ..., P_n$, under the condition that
\begin{equation}
\label{fact}
\mathcal{F}_{\rho} \lbrace P_n,..,P_{j+1} ,P_j \rbrace + \sum_{l \neq l_0} \mathcal{F}_{\rho} \lbrace P_n,..,P_{j+1},P_{j}^{l} \rbrace \approx \mathcal{F}_{\rho} \lbrace P_n,..,P_{j+1} \rbrace
\end{equation}

\begin{definition}
Dephasing
\end{definition}
We say that the equipment $E$ induces ``dephasing'' in a measurement of the quantity $a_{t_j}^{j}$ pertaining to the system $S$, given that $S\vee E$ is prepared in an entangled state $\rho$ before $a_{t_j}^{j}$ is measured and quantities $a_{t_{j+1}}^{j+1},...,a_{t_n}^{n}$ are measured afterwards, if  
\begin{equation}
\label{dephasing}
\rho(QQ^{*})\approx\rho(P_jQQ^{*}P_j) + \sum_{l \neq l_0} \rho(P_{j}^{l}QQ^{*}P_{j}^{l}),
\end{equation}
where $Q=\prod_{k=j+1}^{n}f_{k}(a_{t_k}^{k})$, with $f_k$ an arbitrary continuous function, for $k=j+1,...,n$.

 Dephasing implies that if $S\vee E$ is prepared in the state $\rho$ before $a_{t_j}^{j}$ is measured then interference terms between $P_j$ and complementary possible events $P_j^{l}$, $l\neq l_{0}$, very nearly vanish when quantities $a_{t_j+1}^{j+1},...,a_{t_n}^{n}$ are measured subsequently. Thus, the possible events  $P_{j}^{l}, l=1,...,K_j$, \textit{mutually exclude each other} (for all practical purposes), given those future measurements. Hence, one of these possible events will be observed in a measurement of $a_{t_j}^{j}$.\\
It deserves to be noted that dephasing may be wiped out if the delay between the measurement of $a_{t_j}^{j}$ and subsequent measurements becomes very large or if appropriate measurements (manipulations) on the equipment $E$ are made $after$ the measurement of $a_{t_j}^{j}$, which may lead to a ``disentanglement'' of $E$ and $S$ and, hence, may lead to a re-emergence of interference terms.
 
Next, we attempt to clarify what is meant by decoherence. For this purpose, we introduce a $C^{*}-$algebra $\mathcal{D}_S$, which is generated by the operators
 $$ \lbrace b \mid b= \prod_{k} \alpha_{t_k,t_0} (b^k) , \text{ with } b^k \in \mathcal{A}_S  \rbrace,$$
and a one-parameter group of time-translation automorphisms, $\tau : (\mathcal{D}_S, \mathbb{R}) \rightarrow \mathcal{D}_S$, which is defined by
$$\tau_{t} (b):= \prod_{k} \alpha_{t_k+t,t_0} (b^k), $$
for $b:= \prod_{k} \alpha_{t_{k},t_{0}}(b^k)$.\\
\begin{definition}
Decoherence
\end{definition}
We say that the equipment $E$ induces ``decoherence'' in a measurement of the quantity $a_{t_j}^{j}$ pertaining to the system $S$ if and only if, for all $b \in \mathcal{D}_S$,

\begin{equation}
\label{deco}
 \left[  a^{j}_{t_j}, \tau_t(b)\right] \overset{w}{\rightarrow} 0
\end{equation}
as $t \rightarrow \infty$, in the thermodynamic limit, $\vert\Lambda\vert \rightarrow \infty$, of the equipment $E$. (In (\ref{deco}), $``\overset{w}{\rightarrow}" $ denotes a weak limit.) \\\\
$Remarks.$
\begin{enumerate}[(i)]
\item Decoherence, in the sense of Eq. (\ref{deco}), is only meaningful if the equipment $E$ used in the measurement of the quantity $a_{t_j}^{j}$ is macroscopically large (i.e., in the thermodynamic limit of $E$). \item In contrast to dephasing, decoherence cannot be undone, anymore, because Eq. (\ref{deco}) is assumed to hold for $arbitrary$ operators $b\in\mathcal{D}_{S}$ and $every$ continuous linear functional on the algebra $\mathcal{B}_{S\vee E}$, (hence, for every state of $S \vee E$).
\end{enumerate}
Clearly, decoherence implies dephasing.  We trivially have that

\begin{eqnarray*}
\omega(\tau_t(b)) &=& \omega \left( \left(\sum_{l=1}^{K_j} P_{j}^{l} \right) \tau_t(b) \left(\sum_{m=1}^{K_j} P_{j}^{m} \right) \right) = \sum_{l,m=1}^{K_j}\omega \left(  P_{j}^{l}  \tau_t(b)  P_{j}^{m} \right), 
\end{eqnarray*}
where $\omega$ is a state on $\mathcal{B}_{S\vee E}$.
Moreover, by Eq. (\ref{deco}), $\phi ( \left[ a^{j}_{t_j}, \tau_t(b)\right]) \rightarrow 0$, as $t\rightarrow\infty$, for every continuous linear functional $\phi$ on $\mathcal{B}_{S\vee E}$. Consequently,
\begin{eqnarray*}
\omega( P_{j}^{l}  \left[  a^{j}_{t_j},  \tau_t(b) \right]  P_{j}^{m} )&=& \left(\alpha_{l}- \alpha_m \right) \omega( P_{j}^{l}   \tau_t(b)  P_{j}^{m})\\
&\underset{ t, \vert\Lambda\vert  \rightarrow \infty}{\rightarrow}  & 0
\end{eqnarray*}
because $\omega( P_{j}^{l} \text{ . }  P_{j}^{m} )$ is a continuous linear (but not necessarily positive) functional on $\mathcal{B}_{S\vee E}$.
For $l \neq m$, it follows that  $ \omega( P_{j}^{l}   \tau_t(b)  P_{j}^{m}) \underset{ t, \vert\Lambda\vert  \rightarrow \infty}{\rightarrow}  0$. Thus,
$$\omega(\tau_t(b))  \underset{ t, \vert\Lambda\vert  \rightarrow \infty}{\rightarrow}  \sum_{l=1 }^{K_j} \omega( P_{j}^{l}   \tau_t(b)  P_{j}^{l}) $$
which implies (\ref{fact}), provided $t_{j+1}-t_j$ is large enough.
Dephasing, too, implies (\ref{fact}), but decoherence implies the stronger statement that, in (\ref{fact}), $equality$ holds in appropriate limits, and under suitable assumptions. The importance of dephasing and decoherence in a quantum theory of experiments is that they represent mechanisms that render complementary possible events \textit{mutually exclusive} in an actual experiment, meaning that, for all practical purposes, one of them $will$ happen. Without such mechanisms, it would be impossible to say what one means by ``measuring a physical quantity'' or by ``a possible event to materialize'' (i.e., by a possible event to become a fact).

There is a truly vast literature concerning more or less concrete, more or less realistic models of (dephasing and) decoherence, treated with more or less mathematical precision. One may argue that Schr\"odinger invented the concept of decoherence in connection with his thought experiment on ``Schr\"odinger's cat'' \cite{Schroedinger}. An early contribution towards clarifying this concept appeared in \cite{DeWitt}. Among the first discussions of concrete models of decoherence is the one in \cite{Hepp2}. Obviously, a thorough review of these matters goes beyond the scope of this note; but see e.g. \cite{BL2} or \cite{Schloss}. It will be taken up elsewhere.

\appendix
\appendixpage
\addappheadtotoc
\section{Klyachko's theorem}
\label{A1}
\begin{theorem}(Klyachko)
\label{Klya}
The following conditions are equivalent.
\begin{itemize}
\item The density matrix $\rho$ on $H_1 \otimes H_2$ $(dimH_i = n_i < \infty)$, with spectrum $\underline\lambda$, has $\rho_1$ on $H_1$, with spectrum $\underline\lambda^{(1)}$, and $\rho_2$ on $H_2$, with spectrum $\underline\lambda^{(2)}$, as its marginals.
\item The spectra $\underline\lambda$, $\underline\lambda^{(1)}$ and $\underline\lambda^{(2)}$ satisfy the inequalities
$$ \sum_{i=1}^{n_1} a_i \lambda^{(1)}_{u(i)} +  \sum_{j=1}^{n_2} b_j \lambda^{(2)}_{v(j)} \leq \sum_{k=1}^{n_1 n_2} (a +b)^{\downarrow}_{k} \lambda _{w(k)}$$
for arbitrary non-increasing sequences $a_{n_1} \leq...\leq a_{1}$ and  $b_{n_2} \leq...\leq b_{1}$, with $\sum_{i=1}^{n_1} a_i= \sum_{j=1}^{n_2} b_j=0$, and, for every  permutation, $u$, of $\lbrace 1,...,n_1 \rbrace$, every permutation, $v$, of $\lbrace 1,...,n_2 \rbrace$ and every permutation, $w$, of $\lbrace 1,..., n_1 n_2 \rbrace$ with the property that the Schubert coefficients $c_{w}^{uv} (a,b)$ are non-zero, $(a +b)^{\downarrow}$ is the sequence $a_i+b_j$, arranged in non-increasing order.
\end{itemize}
\end{theorem}

The proof of Klyachko's theorem involves intersection theory for flag varieties and Schubert calculus and is far too sophisticated to be sketched in this  review (see, e.g. \cite{Manivel}). For a somewhat less ambitious introduction to the problem of quantum marginals the reader may consult the work quoted in $\cite{Daftuar}$.

The number of inequalities appearing in Theorem \ref{Klya}  grows very fast in the  dimensions of the spaces  $H_1$ and $H_2$. For instance, for $n_1=2$ and $n_2=4$, one has already 234 inequalities. The usefulness of Theorem \ref{Klya} for concrete purposes of physics thus appears to be rather limited. Moreover, we are not aware of interesting generalizations of this theorem to infinite-dimensional Hilbert spaces.

The following special case is elementary: If the state $\rho$ of the composed system is a $pure$ state (i.e., $\rho$ is given by a unit vector in $H_1 \otimes{H}_{2}$) then its marginals, $\rho_1$ and $\rho_2$, are isospectral ($\underline\lambda_1 = \underline\lambda_2$) but by no means necessarily $pure$. The states $\rho_1$ and $\rho_2$ are pure if and only if $\rho$ is separable. Thus ``knowledge'' of the precise state of the composed system does not, in general, imply that the state on the algebra of possible events in a subsystem is pure, too - in contrast to the situation encountered in realistic theories. This observation offers a way to understand why the probabilities of sequences of events predicted by a pure state of a quantum system do not, in general, satisfy 0--1 laws; or, put differently, quantum systems are not, in general, deterministic.

\section{Tsirelson's work on Bell's inequalities}
\label{A2}
We consider two systems, $S_1$ and $S_2$, that do not interact with each other and the composed system $S:=S_1 \vee S_2$. The following notations and definitions are used.
\begin{itemize}
\item $H_1$, $H_2$ are the state spaces (separable Hilbert spaces) of $S_1$ and $S_2$, respectively. 
\\The category of Hilbert spaces is denoted by $\mathcal{H}$.
\item  We define families, $\mathcal{D}_{H_i}$, of operators on $H_i$ by
$$\mathcal{D}_{H_i}:= \lbrace A \in B(H_i) \mid A^{*}=A, \vert \vert A \vert\vert \leq 1\rbrace, \text{  } i=1,2.$$
\item We restrict our attention to normal states given by density matrices $\rho$, i.e., positive trace-class operators of trace $1$, acting on $H:=H_1 \otimes H_2$. We denote the convex set of density matrices on $H$ by $\mathcal{S}_{H}$.
\item The family of probability spaces, $(\Omega, \mu)$, is called $Prob$. The set of real bounded random variables on $(\Omega, \mu)$, with absolute value bounded above by $1$, is denoted by $\mathcal{D}^{\mathbb{R}}_{\Omega,\mu}$
\end{itemize}

We fix a pair $(K,L)$ of natural numbers. In the following definitions, $k \in \lbrace 1,...,K \rbrace$ and $l \in \lbrace 1,...,L \rbrace$.
\begin{definition} Quantum correlation  matrices
$$\mathcal{M}^{K,L}_{Q}:= 
\lbrace  \Gamma \in M_{K\times{L}}(\mathbb{R}) \mid \Gamma_{kl}= tr(\rho A_k \otimes B_l), \rho \in \mathcal{S}_{H_{1} \otimes{H_2}},  A_k \in \mathcal{D}_{H_1},  B_l \in \mathcal{D}_{H_2}, H_1, H_2 \in \mathcal{H}\rbrace$$
\end{definition}
A subspace of the set, $\mathcal{M}^{K,L}_{Q}$, of quantum correlation matrices is the set of classical correlation matrices, $\mathcal{M}^{K,L}_{C}$, for which the operators $A_k$, $k = 1,..., K$, and $B_l$, $l = 1,..., L$, all commute, (i.e., may represent physical quantities of a realistic system).
\begin{definition} 
Classical correlations  matrices
$$\mathcal{M}^{K,L}_{C}:= \lbrace  \Gamma \in \mathcal{M}^{K,L}_{Q} \mid  \forall (k,k') \in \lbrace 1,...,K \rbrace^{\times2} ,  A_{k} A_{k'}=A_{k'} A_k ; \forall( l,l')\in \lbrace 1,...,L \rbrace^{\times2}, B_{l} B_{l'}=B_{l'} B_l\rbrace$$
\end{definition}
We list some properties of  $\mathcal{M}^{K,L}_{Q}$ and $\mathcal{M}^{K,L}_{C}$. In the following, $(.,.)$ denotes the scalar product on $\mathbb{R}^{n}$. Furthermore, the set of  p-tuples of $p$ unit  vectors in $\mathbb{R}^{n}$ is denoted by $T_{p,n}$.

{\begin{lemma}
\label{Tsitsi}
The spaces $\mathcal{M}^{K,L}_{Q}$ and $\mathcal{M}^{K,L}_{C}$ can be characterized as follows.
\begin{enumerate}[(a)]
  \item An equivalent description of $\mathcal{M}^{K,L}_{Q}$
 $$\mathcal{M}^{K,L}_{Q}=  \lbrace  \Gamma \in M_{K\times L}(\mathbb{R}) \mid  \Gamma_{kl}=  (x_k,y_l),  \text{ where }   (x_1,...,x_K, y_1,...,y_L) \in T_{K+L,K+L} \rbrace$$
  \item An equivalent description of $\mathcal{M}^{K,L}_{C}$
$$\mathcal{M}^{K,L}_{C}:= 
\lbrace  \Gamma \in \mathcal{M}^{K,L}_{Q} \mid 
\Gamma_{kl}= \int_{\Omega} a_k(\omega) b_l(\omega) d\mu(\omega), \text{ }  (\Omega,\mu) \in \text{Prob, }
   a_1,...,a_K,b_1,...,b_L \in \mathcal{D}^{\mathbb{R}}_{\Omega,\mu} \rbrace$$

  \item The set $ \mathcal{M}^{K,L}_{Q}$ is a convex compact subset of the set $ M_{K\times{L}}(\mathbb{R})$ of real $K\times{L}$ matrices; $ \mathcal{M}^{K,L}_{C}$ is a convex polytope in $M_{K\times{L}}(\mathbb{R})$.
\end{enumerate}
\end{lemma}

Convex polytopes are defined as follows \cite{Grunbaum}. 
\begin{definition}
Convex polytopes in $\mathbb{R}^{n}$
\end{definition}
A convex polytope, $\mathcal{P}$, in $\mathbb{R}^{n}$ is a convex compact set with a finite number of extreme points. Equivalently, a convex polytope is the convex hull of a finite set.

The following lemma follows directly from this definition.
\begin{lemma}
A convex polytope $\mathcal{P}$ in $\mathbb{R}^{n}$ can be written as a finite intersection of closed halfspaces $H_{y^{(i)},\alpha_i}$, where $H_{y^{(i)},\alpha_i}:=\lbrace x \in \mathbb{R}^{n} \mid (x,y^{(i)}) \leq \alpha_i \rbrace$, with $\alpha_i \in \mathbb{R}$ and $0\neq y^{(i)} \in \mathbb{R}^{n} $.
\end{lemma}

The proof of this lemma can be found in standard text books on convex polytopes. A proof of Lemma 6 may be found in the work of Tsirelson. Though elementary, it is too long to be reproduced here.

To state Tsirelson's main result on Bell's inequalities, we need a simple version of an inequality on tensor products due to Grothendieck.
\begin{theorem} (Grothendieck)
\label{Grot}
Let $n \in \mathbb{N}$, and  let $a_{ij}$ be a real $n \times n$ matrix. If, for any $\underline{s}=(s_1,...,s_n), \underline{t}=(t_1,...,t_n) \in \mathbb{R}^{n}$,
$$\vert \sum_{i,j=1}^{n} a_{ij} s_i t_{j}\vert\leq \max_{i} \vert s_i \vert \max_{j} \vert t_j \vert $$
then, for an arbitrary set of vectors $x_i, y_j  \in H$, where $H$ is a Hilbert space,
$$ \vert  \sum_{i,j=1}^{n} a_{ij} (x_i,y_j) \vert  \leq K_n \max_{i} \vert \vert x_i \vert \vert \max_{j} \vert \vert y_j \vert \vert, $$
for some constant $1< K_n \in \mathbb{R}_{+}$.
The smallest constant larger or equal to $K_n$, for all $n \in \mathbb{N}$, is denoted by $K_{G}$ and is called Grothendieck's constant.
\end{theorem}
For a proof see, e.g., \cite{Pisier}.

\begin{theorem} (Tsirelson)
\label{tsi}
Let $\Gamma \in \mathcal{M}^{K,L}_{Q}$. Then $K_{G}^{-1} \Gamma \in  \mathcal{M}^{K,L}_{C}.$
\end{theorem}
\begin{proof}
As $\mathcal{M}^{K,L}_{C}$ is a convex polytope in $\mathbb{R}^{KL}$, it can be written as a finite intersection of $d$ closed halfspaces $H_{y^{(m)},\alpha_m}$, and $\Gamma \in \mathcal{M}^{K,L}_{C}$ iff $(\Gamma, y^{(m)}) \leq \alpha_m$, for all $m=1,...,d$. If $\Gamma \in \mathcal{M}^{K,L}_{C}$, then $- \Gamma \in \mathcal{M}^{K,L}_{C}$, which follows from Lemma \ref{Tsitsi} by changing $a_k \mapsto -a_{k}$. Thus, for $\Gamma \in \mathcal{M}^{K,L}_{C}$, we also have that $(\Gamma, y^{(m)}) \geq -\alpha_m$, i.e., if  $H_{y^{(m)},\alpha_m}$ appears in the intersection defining $\mathcal{M}^{K,L}_{C}$ then so does $H_{-y^{(m)},\alpha_m}$. Grouping such symmetric halfspaces into pairs and labeling these pairs, we conclude that$\Gamma \in \mathcal{M}^{K,L}_{C}$ iff $\vert (\Gamma, y^{(m)})  \vert \leq \vert \alpha_m \vert$, for $m=1,...,\frac{d}{2}$. We divide each $y^{(m)}$ by $\alpha_m $ and denote the resulting vector by $a^{(m)}$. Then $\Gamma \in \mathcal{M}^{K,L}_{C}$ iff
$$\vert (\Gamma,a^{(m)}) \vert = \vert \sum_{k,l=1}^{K,L} \Gamma_{kl} a^{(m)}_{kl}  \vert \leq 1 $$
These inequalities hold, in particular, for the special set of matrices $\Gamma$ for which $\Gamma_{kl}= s_k t_l$, where $\vert s_k \vert, \vert t_l \vert \leq 1$. Such matrices are classical correlation matrices coming from constant random variables $s_k$,$t_l$.  For arbitrary $ \vert s_k \vert , \vert t_l \vert \leq 1$, the $K \times L$ matrix $a^{(m)}$ satisfies
$$\vert \sum_{k,l=1}^{K,L} s_k t_l  a^{(m)}_{kl} \vert \leq 1,$$
 for any $m=1,...,d/2.$ 
Let $n:=\max(K,L)$ and enlarge $a^{(m)}_{kl} \in M_{K\times L}(\mathbb{R})$ to $\tilde{a}^{(m)}_{ij} \in M_{n}(\mathbb{R})\equiv M_{n\times n}(\mathbb{R}) $ by setting to $0$ the added matrix elements. Then, for any $ \vert s_i \vert, \vert t_j \vert \leq 1$, the $n \times n$ matrix $\tilde{a}^{(m)}$ satisfies
$$\vert \sum_{i,j=1}^{n}  s_i  t_j  \tilde{a}^{m}_{ij}  \vert \leq 1$$ 

Let $\Gamma^{q} \in  \mathcal{M}^{K,L}_{Q}$. By Lemma \ref{Tsitsi}, $\Gamma^{q}_{kl}=(x_k,y_l)$, for some vectors $x_k, y_l$ in $\mathbb{R}^{K+L}$ of norm  one. We enlarge this family of vectors to a total of $2n$ vectors by setting the added ones to 0. According to Theorem \ref{Grot}, $$\vert \sum_{i,j=1}^{n}  (x_i,  y_j)  \tilde{a}^{(m)}_{ij}  \vert \leq K_n,$$
for any $m=1,...,d/2,$
i.e.,
$$\vert \sum_{k,l=1}^{K,L}  \frac{1}{K_n} \Gamma^{q}_{kl}  {a}^{(m)}_{kl}  \vert \leq 1$$ for all $m=1,...,d/2$. In other words, $K_{n}^{-1}\Gamma^{q}_{kl}$ lies in the convex polytope $\mathcal{M}^{K,L}_{C}$.
\end{proof}
Theorem \ref{tsi} does not imply that the set of quantum correlation matrices is strictly larger than the set of classical ones.  This can, however, easily be shown in special cases. A well known example is the CHSH inequality, see section 3.2.6. We set $(K,L)=(2,2)$; $A_1$,$A_2$ are quantum observables of $S_1$ and $B_1$,$B_2$   quantum observables of $S_2$, with dim$(H_i)$ = 2, for $i=1,2$. In classical theories, these observables are random variables denoted by $a_1,a_2$ and $b_1,b_2$, with $ \vert a_i\vert \leq 1$ and $ \vert b_i\vert \leq 1$.  Classical correlation matrices in $\mathcal{M}^{2,2}_{C}$ are $2 \times 2$ real matrices. Extreme points of  $\mathcal{M}^{2,2}_{C}$  are correlation matrices $\Gamma^{ex}_{kl}=a_k b_l$, with $a_k=\pm 1$, $b_l =\pm 1$. The  number of extreme points of $\mathcal{M}^{2,2}_{C}$ is equal to 8, four of them being given by
$$\Gamma^{1}_{kl}:=\left( \begin{array}{cc}
1&  1 \\
1 &  1
\end{array} \right), \text{ }
\Gamma^{2}_{kl}:=\left( \begin{array}{cc}
1&  -1 \\
1 & - 1
\end{array} \right), \text{ }
\Gamma^{3}_{kl}:=\left( \begin{array}{cc}
-1&  -1 \\
\ \ 1 & \  \ 1
\end{array} \right), \text{ }
\Gamma^{4}_{kl}:=\left( \begin{array}{cc}
\ \ 1&  -1 \\
-1 &  \ \ 1
\end{array} \right)
$$
and the remaining four by multiplying the first four matrices by $-1$. The symmetric polytope $\mathcal{M}^{2,2}_{C}$ has $8$ three-dimensional faces, and every classical correlation matrix, $\Gamma^{c}$, satisfies a set of $4$ inequalities of the form
$$\vert \sum_{ij=1}^{2} a^{(m)}_{ij} \Gamma^{c}_{ij} \vert \leq \vert \alpha_{m} \vert,$$
where $m=1,...,4$.
The independent inequalities
$$\vert  \Gamma^{c}_{11}+\Gamma^{c}_{12}+\Gamma^{c}_{21}+\Gamma^{c}_{22}-2 \Gamma^{c}_{kl} \vert \leq  2 $$
are satisfied by every extreme point of the polytope, every inequality is saturated by four of them. They  entirely characterize the set $\mathcal{M}^{2,2}_{C}$. We can write them in the form $\vert \sum_{ij=1}^{2} a^{(m)}_{ij} \Gamma^{c}_{ij} \vert \leq \vert \alpha_{m} \vert$, choosing
$$a^{(m)}:=\left( \begin{array}{cc}
 1&  1 \\
1 &  1
\end{array} \right) - 2 E^{m},$$
$\alpha_m=2$, and
$$E^{1}:=\left( \begin{array}{cc}
 1&  0 \\
0 &  0
\end{array} \right) , \text{ }
E^{2}:=\left( \begin{array}{cc}
 0&  1 \\
0 &  0
\end{array} \right) , \text{ }
E^{3}:=\left( \begin{array}{cc}
 0&  0 \\
1 &  0
\end{array} \right) , \text{ }
E^{4}:=\left( \begin{array}{cc}
 0&  0 \\
0 &  1
\end{array} \right) 
$$
These inequalities are violated by quantum correlation matrices. Let $\Gamma^{q} \in \mathcal{M}^{2,2}_{Q}$. According to Lemma \ref{Tsitsi}, there exist unit vectors $x_1,x_2,y_1,y_2 \in \mathbb{R}^{4}$ such that $ \Gamma^{q}_{kl}=(x_k,y_l)$. Without loss of generality, we consider one of the four inequalities, e.g.,
\begin{eqnarray*}
\sum_{i,j=1}^{2} a_{ij}^{(1)}
\Gamma^{q}_{ij}&=&\Gamma^{q}_{12}+\Gamma^{q}_{21}+\Gamma^{q}_{22}-\Gamma^{q}_{11}\\
&=&(x_1,y_2-y_1)+(x_2,y_1+y_2)
\end{eqnarray*}
By the Cauchy-Schwarz inequality,
$$\vert (x_1,y_2-y_1)+(x_2,y_1+y_2) \vert \leq  \vert \vert y_2-y_1  \vert \vert +  \vert   \vert  y_2+y_1 \vert \vert=\sqrt{2-2(y_1,y_2)}+\sqrt{2+2(y_1,y_2)}  $$
But, for any $x \leq 1$, 
$$(\sqrt{1-x}+\sqrt{1+x})^{2} =2+2\sqrt{1-x^{2}} \leq 4,$$ 
i.e.,
$$\vert (x_1,y_2-y_1)+(x_2,y_1+y_2) \vert \leq 2 \sqrt{2}  $$

These inequalities are saturated for $y_1 \perp y_2$ and $x_{1}=\frac{1}{\sqrt{2}}(y_2-y_1)$, $x_{2}=\frac{1}{\sqrt{2}}(y_2+y_1)$. In most textbooks on quantum mechanics, a system of two spin 1/2 particles is considered to physically interpret these quantum correlation matrices; see e.g. \cite{Haroche}.
For $(K,L)=(2,2)$, the constant $K_{2}$ has the value $\sqrt{2}$. Moreover, quantum correlation matrices  violating the classical inequalities lie outside of the classical polytope, i.e., $\mathcal{M}^{2,2}_{C} \subsetneq \mathcal{M}^{2,2}_{Q}$.

\section{Proofs of Lemmas \ref{lem1} and \ref{lem2}, subsection \ref{4.5}}
\label{A3}
We first prove Lemma $\ref{lem2}$.
\begin{proof}.
We denote the spectrum of an operator $A$ by $\sigma(A)$ and set $ f(x):=x^{2}-x$.
Because $P$ is a selfadjoint bounded operator, $\sigma(f(P))=f(\sigma(P))$, and $\vert \vert f(P) \vert \vert=\sup_{\lambda \in \sigma(P)} \vert \lambda^{2} -\lambda \vert <\epsilon$, by hypothesis. We consider the polynomials $Q_{\epsilon}(X):=X^{2}-X-\epsilon$ and $Q'_{\epsilon}(X):=-X^{2}+X-\epsilon$. The real roots of these polynomials are given by
$$x_{\pm}(\epsilon)=\frac{1 \pm \sqrt{1+4 \epsilon}}{2}, \text{ } x'_{\pm}(\epsilon)=\frac{1 \pm \sqrt{1-4 \epsilon}}{2},$$
respectively.
Denoting $$\Delta_{0}:=\left] \frac{1-\sqrt{1+4 \epsilon}}{2},\frac{1-\sqrt{1-4 \epsilon}}{2} \right[$$ and  $$\Delta_{1}:=\left]\frac{1 + \sqrt{1-4 \epsilon}}{2},\frac{1 +\sqrt{1+4 \epsilon}}{2} \right[$$ we find that $\sigma(P) \subset \Delta_{0} \cup \Delta_{1}$. \\
According to the spectral theorem,
$$P=\int_{\sigma(P)} \lambda \text{ } dE_{P}(\lambda)$$
We define $$\hat{P}:=\int_{\sigma(P) \cap \Delta_{1}} dE_{P}(\lambda)$$
Clearly, $\hat{P}$ is an orthogonal projection. Moreover, 
$$\hat{P}-P =\int_{\sigma(P) \cap \Delta_{1}} (\lambda -1 ) \text{ }  dE_{P}(\lambda) + \int_{\sigma(P) \cap \Delta_{0}} \lambda \text{ } dE_{P}(\lambda)$$
For $\lambda \in \sigma(P) \cap \Delta_{1}$, $\vert \lambda -1 \vert <\frac{2 \epsilon}{1+\sqrt{1-4 \epsilon}} $,  and, for $\lambda \in \sigma(P) \cap \Delta_{0}$, 
 $\vert \lambda \vert <\frac{2 \epsilon}{1+\sqrt{1-4 \epsilon}} $. Consequently, every element in the spectrum of $\hat{P}-P$ is smaller, in absolute value, than $\frac{2 \epsilon}{1+\sqrt{1-4 \epsilon}} $, and thus
$$\vert \vert \hat{P}-P \vert \vert \leq \frac{2 \epsilon}{1+\sqrt{1-4 \epsilon}} < 2 \epsilon.$$
\end{proof}

Next, we turn to the proof of Lemma $\ref{lem1}$.

\begin{proof}
We construct the orthogonal projections $\tilde{P}_{i}$ inductively. 

We first consider the case where $n=2$. Let $P_{1}$ and $P_{2}$ be two orthogonal projections satisfying the hypotheses of Lemma \ref{lem2}. We set $\tilde{P}_{2}:=P_{2}$. To construct $\tilde{P}_1$, we define operators $Q:=\tilde{P}_{2}^{\perp} P_1 \tilde{P}_{2}^{\perp} $ and $Q{'}:=\tilde{P}_{2} P_1 \tilde{P}_{2} $. Clearly, $Q$ and $Q{'}$ are selfadjoint bounded operators. Moreover, 
\begin{eqnarray*}
Q^{2}-Q&=&\tilde{P}_{2}^{\perp} P_1 \tilde{P}_{2}^{\perp}  P_1 \tilde{P}_{2}^{\perp}- \tilde{P}_{2}^{\perp} P_1 \tilde{P}_{2}^{\perp} \\
&=&\tilde{P}_{2}^{\perp} P_1  \left[\tilde{P}_{2}^{\perp}, P_1  \right]\tilde{P}_{2}^{\perp}
\end{eqnarray*}
and hence
$$\Vert Q^{2}-Q \Vert < \epsilon,$$ 
by hypothesis. The same holds for $(Q{'})^{2}-Q{'}$. According to Lemma \ref{lem1}, there is an orthogonal projection $\hat{Q}$ commuting with $Q$ and an orthogonal projection $\hat{Q'}$ commuting with $Q'$  such that $\Vert \hat{Q}-Q \Vert  < 2 \epsilon $ and   $\Vert \hat{Q'}-Q' \Vert < 2 \epsilon $. 
We define
$$\tilde{P}_{1}:=\hat{Q} \tilde{P}_{2}^{\perp} + \hat{Q'} \tilde{P}_{2}$$
which is easily seen to be a projection commuting with $\tilde{P}_2$; (as noticed at the end of the proof of Lemma \ref{lem1}). Moreover,
\begin{eqnarray*}
\vert \vert P_{1}-\tilde{P}_{1} \vert \vert &\leq&  \vert \vert P_{1}-Q -Q'\vert \vert +\vert \vert Q -\hat{Q}\tilde{P}_{2}^{\perp}  \vert \vert +\vert \vert Q' -\hat{Q'} \tilde{P}_{2} \vert \vert\\
&<&\vert \vert  \tilde{P}_{2} P_1  \tilde{P}_{2}^{\perp}+ \tilde{P}_{2}^{\perp} P_1 \tilde{P}_{2}\vert \vert + 4 \epsilon\\
&<& 6 \epsilon
\end{eqnarray*}
using that $ \tilde{P}_{2} P_1  \tilde{P}_{2}^{\perp}+ \tilde{P}_{2}^{\perp} P_1 \tilde{P}_{2}=\tilde{P}_{2} \left[P_1 , \tilde{P}_{2}^{\perp} \right]+ \left[ \tilde{P}_{2}^{\perp}, P_1\right] \tilde{P}_{2}$. We have thus constructed two commuting projections $\tilde{P}_1$ and $\tilde{P}_2$ with the properties claimed to hold in Lemma \ref{lem2}, with $C_{2}=6$. It follows that $\tilde{P}_1\tilde{P}_2\tilde{P}_1$ is an orthogonal projection, as well.
 
Let $n\in\mathbb{N}$ and let $j>1$ be an integer smaller than n. We suppose that we have already constructed projections $ \tilde{P}_{j},...,\tilde{P}_n$, starting from $P_n$, 
such that $\left[\tilde{P}_{k},(\Pi_{i=k+1}^{n}\tilde{P}_i)( \Pi_{i=n}^{k+1} \tilde{P}_i) \right]=0$, for $k=j,...,n-1$, and $\vert \vert P_{k}-\tilde{P}_{k} \vert \vert <C_{n-k+1} \epsilon$, for $k=j,...,n$,  with $\epsilon< \frac{1}{4 (4\sum_{i=j}^{n} C_{n-i+1}+1)}$. We proceed to construct a projection $\tilde{P}_{j-1}$ close to $P_{j-1}$ and commuting with the operator $\tilde{H}_{j-1}:=(\Pi_{i=j}^{n}\tilde{P}_i)(\Pi_{i=n}^{j} \tilde{P}_i)$, using the ideas used above to prove the lemma in the special case where $n=2$. Then $\tilde{H}_{j-1}$ takes the role of $P_2$ and $P_{j-1}$ the role of $P_1$ in the argument to prove the special case where $n=2$. Indeed, define $Q_{j-1}=\tilde{H}_{j-1} P_{j-1} \tilde{H}_{j-1}$, $Q'_{j-1}=\tilde{H}_{j-1}^{\perp} P_{j-1} \tilde{H}_{j-1}^{\perp}$.
\begin{equation}
Q^{2}_{j-1}-Q_{j-1}=\tilde{H}_{j-1}^{\perp} P_{j-1} \left[ \tilde{H}_{j-1}^{\perp}, P_{j-1}\right]\tilde{H}_{j-1}^{\perp}
\end{equation}
and thus,
\begin{eqnarray*}
\vert \vert Q^{2}_{j-1}-Q_{j-1}\vert \vert& \leq& \vert \vert  \left[ \tilde{H}_{j-1}^{\perp}, P_{j-1}\right]\vert \vert \leq \vert \vert   \left[ \tilde{H}_{j-1} -H_{j-1}, P_{j-1}\right] \vert \vert  +\vert \vert  \left[ H_{j-1}, P_{j-1}\right] \vert \vert \\
&<& 2 \vert \vert  \tilde{H}_{j-1} -H_{j-1} \vert \vert + \epsilon\\
&<& 4 \epsilon \sum_{i=j}^{n}C_{n-i+1}  + \epsilon
\end{eqnarray*}
where the last inequality follows from the use of $\tilde{P}_{k}=(\tilde{P}_{k}-P_{k})+ P_{k}$ in the expression of $\tilde{H}_{j-1}$. The same holds for $Q'_{j-1}$. Since we have assumed that  $\epsilon< \frac{1}{4 (4\sum_{i=j}^{n} C_{n-i+1}+1)}$, we can find $\hat{Q}_{j-1}$ and $\hat{Q'}_{j-1}$ obeying the conditions of lemma \ref{lem1}, i..e. $\hat{Q}_{j-1}$ and $\hat{Q'}_{j-1}$ are orthogonal projections commuting with $Q_{j-1}$, $Q'_{j-1}$, respectively, with $\Vert\hat{Q}_{j-1} -Q_{j-1} \Vert <  8 \epsilon (\sum_{i=j}^{n}C_{n-i+1})  + 2 \epsilon$. Defining
$$\tilde{P}_{j-1}:=\hat{Q}_{j-1}   \tilde{H}_{j-1}^{\perp}+ \hat{Q'}_{j-1}   \tilde{H}_{j-1}$$
one finds that 
$$\vert \vert \tilde{P}_{j-1}-P_{j-1}\vert \vert <6 \epsilon \left( 4  \sum_{i=j}^{n}C_{n-i+1}  +1 \right) $$
i.e., $C_{n-j+2}=6\left( 4  \sum_{i=j}^{n}C_{n-i+1}  +1 \right)$.
\end{proof}

\bibliographystyle{plain}
\nocite{*}
\bibliography{QM2-1}

\end{document}